\documentclass[12pt]{article}
\usepackage{amsmath}
\usepackage{graphicx,psfrag,epsf}
\usepackage{enumerate}
\usepackage{natbib}
\usepackage{url} 

\newcommand{\blind}{1}

\usepackage{natbib}

\usepackage{enumitem}  
\usepackage{pdfpages}
\usepackage{makecell}
\usepackage{fullpage}
\usepackage{amsmath}
\usepackage{amssymb}
\usepackage{amsthm}

\usepackage{subfigure}
\usepackage{footnote}
\usepackage{tabu}
\usepackage{tikz}
\usepackage{booktabs}

\usepackage{url}

\usepackage{multicol}
\usepackage{xcolor}
\usepackage{algorithm}
\usepackage[ruled,vlined,algo2e]{algorithm2e}

\usepackage[colorlinks, linkcolor=blue, anchorcolor=blue, citecolor=blue]{hyperref}

\DeclareMathAlphabet{\mathsf}{OT1}{cmss}{m}{n}
\SetMathAlphabet{\mathsf}{bold}{OT1}{cmss}{bx}{n}

\newcommand{\EE}{\mathbb{E}}
\newcommand{\NN}{\mathbb{N}}
\newcommand{\PP}{\mathbb{P}}
\newcommand{\RR}{\mathbb{R}}

\newcommand{\bh}{\mathbf{h}}
\newcommand{\bk}{\mathbf{k}}
\newcommand{\bp}{\mathbf{p}}

\newcommand{\bx}{\mathbf{x}}
\newcommand{\by}{\mathbf{y}}
\newcommand{\bz}{\mathbf{z}}

\newcommand{\cD}{\mathcal{D}}
\newcommand{\cF}{\mathcal{F}}
\newcommand{\cH}{\mathcal{H}}

\newcommand{\cM}{\mathcal{M}}
\newcommand{\cN}{\mathcal{N}}
\newcommand{\cO}{\mathcal{O}}

\newcommand{\EI}{\mathop{\mathrm{EI}}}
\newcommand{\HEI}{\mathop{\mathrm{HEI}}}

\newcommand{\argmin}{\mathop{\mathrm{argmin}}}
\newcommand{\argmax}{\mathop{\mathrm{argmax}}}

\newcommand{\distas}[1]{\mathbin{\overset{#1}{\kern\z@\sim}}}%

\newcommand{\cbl}[1]{{{#1}}}
\providecommand{\norm}[1]{\|#1\|}

\ifx\theorem\undefined
\newtheorem{assumption}{Assumption}
\fi
\ifx\theorem\undefined
\newtheorem{theorem}{Theorem}
\fi
\ifx\lemma\undefined
\newtheorem{lemma}{Lemma}
\fi
\ifx\proposition\undefined
\newtheorem{proposition}{Proposition}
\fi
\ifx\condition\undefined
\newtheorem{condition}{Condition}
\fi
\ifx\conjecture\undefined
\newtheorem{conjecture}{Conjecture}
\fi

\begin{document}

 \def\spacingset#1{\renewcommand{\baselinestretch}%
{#1}\small\normalsize} \spacingset{1}

\if1\blind
{
\title{A hierarchical expected improvement method for Bayesian optimization}
\author{Zhehui Chen\thanks{Joint first author, H. Milton Stewart School of Industrial and Systems Engineering, Georgia Institute of Technology}, Simon Mak\thanks{Joint first author, Department of Statistical Science, Duke University}, C. F. Jeff Wu\thanks{Corresponding author, H. Milton Stewart School of Industrial and Systems Engineering, Georgia Institute of Technology} \footnote{This research is supported by ARO W911NF-17-1-0007, NSF DMS-1914632, NSF CSSI Frameworks 2004571, and NSF DMS 2210729.}}


 \maketitle
} \fi

\if0\blind
{
  \bigskip
  \bigskip
  \bigskip
  \begin{center}
    {\LARGE\bf A hierarchical expected improvement method for Bayesian optimization}
\end{center}
  \medskip
} \fi

\bigskip

\begin{abstract}
The Expected Improvement (EI) method, proposed by \cite{jones1998efficient}, is a widely-used Bayesian optimization method, which makes use of a fitted Gaussian process model for efficient black-box optimization. However, one key drawback of EI is that it is overly greedy in exploiting the fitted Gaussian process model for optimization, which results in suboptimal solutions even with large sample sizes. To address this, we propose a new hierarchical EI (HEI) framework, which makes use of a hierarchical Gaussian process model. HEI preserves a closed-form acquisition function, and corrects the over-greediness of EI by encouraging exploration of the optimization space. We then introduce hyperparameter estimation methods which allow HEI to mimic a fully Bayesian optimization procedure, while avoiding expensive Markov-chain Monte Carlo sampling steps. We prove the global convergence of HEI over a broad function space, and establish near-minimax convergence rates under certain prior specifications. Numerical experiments show the improvement of HEI over existing Bayesian optimization methods, for synthetic functions and a semiconductor manufacturing optimization problem.

\end{abstract}

\noindent%
{\it Keywords:} Black-box Optimization, Experimental Design, Hierarchical Modeling, Gaussian Process, Global Optimization, Uncertainty Quantification.
\vfill
\newpage
\spacingset{1.45} 

\section{Introduction}

Bayesian optimization (BO) is a widely-used optimization framework, which has broad applicability in a variety of problems, including rocket engine design \citep{mak2018efficient}, nanowire yield optimization \citep{dasgupta2008statistical}, and surgery planning \citep{chen2020function}. BO aims to solve the following black-box optimization problem:
\begin{align}\label{eqn:go}
 \bx^* = \argmin_{\bx\in \Omega} f(\bx).
\end{align}
Here, $\bx\in\RR^d$ are the input variables, and $\Omega\subset \RR^d$ is the feasible domain for optimization. The key challenge in \eqref{eqn:go} is that the objective function $f(\cdot): \Omega \rightarrow \RR$ is assumed to be black-box: it admits no closed-form expression, and evaluations of $f$ ({which we presume to be noiseless in this work}) typically require expensive simulations or experiments. For such problems, an optimization procedure should find a good solution to \eqref{eqn:go} given {limited} function evaluations. BO achieves this by first assigning to $f$ a \textit{prior model} capturing prior beliefs on the objective function, then sequentially querying $f$ at points which maximize the \textit{acquisition function} -- the posterior expected utility of a new point. This provides a principled way {to perform} the so-called \textit{exploration-exploitation trade-off} \citep{kearns2002near}: {exploring} the black-box function over $\Omega$, and {exploiting} the fitted function when appropriate for optimization.



Much of the literature on BO can be categorized by (i) the prior stochastic model assumed on $f$, and (ii) the utility function used for sequential sampling. For (i), the most popular stochastic model is the Gaussian process (GP) model~\citep{santner2003design}. Under a GP model, several well-known BO methods have been derived using different utility functions for (ii). These include the expected improvement (EI) method \citep{mockus1978application,jones1998efficient}, the upper confidence bound (UCB) method~\citep{Srinivas:2010:GPO:3104322.3104451}, and the Knowledge Gradient method~\citep{Frazier:2008:KPS:1461633.1461641,scott2011correlated}. Of these, EI is arguably the most popular method, since it admits a simple \textit{closed-form} acquisition function, which can be efficiently optimized to yield subsequent query points on $f$. EI has been subsequently developed for a variety of black-box optimization problems, including multi-fidelity optimization \citep{zhang2018variable}, constrained optimization \citep{feliot2017bayesian}, and parallel/batch-sequential optimization \citep{marmin2015differentiating}.

Despite the popularity of EI, it does have key limitations. One such limitation is that it is too \textit{greedy} \citep{qin2017improving}: EI focuses nearly all sampling efforts near the optima of the \textit{fitted} GP model, and does not sufficiently explore other regions. In terms of the exploration-exploitation trade-off \citep{kearns2002near}, EI over-exploits the fitted model on $f$, and under-explores the domain; this causes the procedure to get stuck in local optima and not converge to any global optimum $\bx^*$ \citep{bull2011convergence}. {In recent work, an effective way to address this greediness is by \textit{integrating uncertainty on model parameters} within the EI acquisition function. \cite{snoek2012practical} proposed a fully Bayesian EI, where GP model parameters are sampled using Markov chain Monte Carlo (MCMC); this incorporates parameter uncertainty within EI via a fully Bayesian framework, which enables improved optimization by encouraging exploration}. \cite{chen2017sequential} proposed a variation of EI under an additive Bayesian model, which encourages exploration by increasing model uncertainty. {However, by integrating parameter uncertainty, a \textit{significant} bottleneck of these methods is that it requires \textit{expensive} MCMC sampling, which can be very costly to integrate within the maximization of the acquisition function.} This computational burden diminishes a key advantage of EI: efficient queries via a \textit{closed-form} criterion.  

Another remedy, proposed by \cite{bull2011convergence}, is to artificially inflate the maximum-likelihood estimator of the GP variance, and use this inflated estimate within the EI acquisition function for sequential sampling. This, in effect, encourages exploration of the black-box function by inflating the uncertainty of the fitted model. The work further employs an ``$\epsilon$-greedy'' modification of the EI, where at each sampling iteration, one selects the next point uniformly-at-random with probability $\epsilon \in (0,1)$. With a sufficiently large $\epsilon$, this again forces the procedure to explore the feasible domain. {While the EI with such modifications (we call this the $\epsilon$-EI later) allows for appealing theoretical properties, such an approach may yield suboptimal optimization performance for black-box problems with limited function evaluations, as we show later in numerical experiments. One reason is that such adjustments, while indeed encouraging exploration, does so in a rather \textit{ad-hoc} manner from a modeling perspective. As such, when evaluations are limited, the $\epsilon$-EI can be suboptimal for striking a good balance between exploration and exploitation, particularly for higher-dimensional problems.}



{To address these limitations, we propose a novel Hierarchical EI (HEI) framework for effective Bayesian optimization with limited function evaluations. Similar to \cite{snoek2012practical} and work thereafter, the HEI integrates parameter uncertainty on model parameters, but does so via the \textit{hierarchical} GP model in \cite{handcock1993bayesian}, which we show yields a \textit{closed-form} acquisition function for Bayesian optimization. This closed-form enables \textit{efficient} sequential sampling, without the need for integrating expensive MCMC samples in optimizing subsequent queries. With this hierarchical framework, we further introduce hyperparameter estimation methods, which allow the HEI to mimic a fully Bayesian optimization procedure (the ``gold standard'' for uncertainty quantification) while avoiding expensive MCMC steps. Under certain prior specifications, we then prove that the HEI indeed converges to a global optimum $\bx^*$ over a broad function space for $f$ and achieves a near-minimax convergence rate. These theoretical guarantees are similar to those achieved by the $\epsilon$-EI \citep{bull2011convergence} and more recent results (see, e.g., \citealp{wynne2020convergence}), but are achieved under a principled hierarchical Bayesian framework, without the need for ad-hoc variance inflation or $\epsilon$-greedy sampling. As such, the proposed HEI provides a principled balance between exploration and exploitation guided by the underlying hierarchical GP model, which in turn yields improved practical performance given limited function evaluations. We finally demonstrate the sample-efficient performance of HEI over existing methods in a suite of numerical experiments and an application to semi-conductor manufacturing. }

We note that a special case of HEI, called the Student EI (SEI), was proposed earlier in \cite{benassi2011robust}.   The proposed HEI has several key advantages over the SEI: the HEI incorporates uncertainty on process nonstationarity, and can mimic a fully Bayesian optimization procedure via hyperparameter estimation. We also show that the HEI has provable global convergence and convergence rates for optimization, whereas the SEI (with the recommended hyperparameter specification) can fail to converge to a global optimum $\bx^*$. Numerical experiments and a semiconductor manufacturing application show the improved performance of the HEI over existing BO methods, including {the SEI approach in \cite{benassi2011robust} and the $\epsilon$-EI approach in \cite{bull2011convergence}}.

The paper is organized as follows: Section~\ref{sec:pro} reviews the GP model and the EI method. Section~\ref{sec:model} presents the HEI method and contrasts it with existing methods. Section~\ref{sec:hyper} provides methodological developments on hyperparameter specification and basis selection. Section~\ref{sec:the} proves the global convergence for HEI and its associated convergence rates.  Sections~\ref{sec:exp} and \ref{sec:app} compare HEI with existing methods in a suite of numerical experiments and for a semiconductor manufacturing problem, respectively. Concluding remarks are given in~Section~\ref{sec:conc}.

%
%
%
%
%
%
%

\section{Background and Motivation}\label{sec:pro}
We first introduce the GP model, then review the EI method and its deficiencies, which will help motivate the proposed HEI method.

\subsection{Gaussian Process Modeling}

We first model the black-box objective function $f$ as the following Gaussian process model:
\begin{align}\label{eqn:uk}
f(\bx)= \mu(\bx) + Z(\bx),\quad  \mu(\bx) = \bp^\top(\bx) \boldsymbol{\beta}, \quad Z(\bx)\sim \mathrm{GP}(0,\sigma^2 K),
\end{align}
where $\mu(\bx)$ is the mean function of the process, $\bp(\bx) = [p_1(\bx),\cdots,p_q(\bx)]^\top$ are the $q$ basis functions for $\mu(\bx)$, and $\boldsymbol{\beta}\in\RR^q$ are its corresponding coefficients. Here, we assume the residual term $Z(\bx)$ follows a stationary Gaussian process prior \citep{santner2003design} with mean zero, process variance $\sigma^2$ and correlation function $K(\cdot,\cdot)$, which we denote as $Z(\bx) \sim \mathrm{GP}(0,\sigma^2 K)$. The model \eqref{eqn:uk} is known as the \textit{universal kriging} (UK) model in the geostatistics literature~\citep{wackernagel1995multivariate}. When \cbl{the modeler does not have prior information on appropriate basis functions to use}, one can employ the so-called \textit{ordinary kriging} (OK) model \citep{olea2012geostatistics}, which sets a constant mean for $\mu(\bx)$, i.e., $p_1(\bx)=1$ and $q=1$. {In what follows, we assume that the kernel $K$ is radial with length-scale parameters $\boldsymbol{\theta}$; this is more formally defined later in Section \ref{sec:the} (see equation \eqref{eq:ker}).}

Suppose {noiseless} function values $y_i = f(\bx_i)$ are observed at inputs $\bx_i$, yielding data $\cD_n = \{(\bx_i,y_i)\}_{i=1}^n$. Let $\by_n = (y_i)_{i=1}^n$ be the vector of observed function values, $\bk_n(\bx) = (K(\bx,\bx_i))_{i=1}^n$ be the correlation vector between the unobserved response $f(\bx)$ and observed responses $\by_n$, $\mathbf{K}_n = {(K(\bx_i,\bx_j))_{i,j=1}^n}$ be the correlation matrix for observed points, and $\mathbf{P}_n = [\bp(\bx_1),\cdots,\bp(\bx_n)]^\top$ be the model matrix for observed points.  
Then the posterior distribution of $f(\bx)$ at an unobserved input $\bx$ has the closed form expression \citep{santner2003design}:
\begin{align}\label{eqn:model}
[f(\bx)\big|\cD_n] \sim \cN\Big( \hat{f}_n(\bx), \sigma^2 s_n^2(\bx)\Big)~.
\end{align}
Here, the posterior mean $\hat{f}_n(\bx)$ is given by:
\begin{align}
\label{eqn:hat}
    \hat{f}_n(\bx)=\bp^{\top}(\bx)\hat{\boldsymbol{\beta}}_n + \bk^{\top}_n(\bx) \mathbf{K}_n^{-1} \left(\by_n - \mathbf{P}_n \hat{\boldsymbol{\beta}}_n \right),
\end{align}
the posterior variance $\sigma^2 s_n^2(\bx)$ is given by:
\begin{align}
\label{eqn:post_var}
    \sigma^2 s_n^2(\bx) = \sigma^2 \left(K(\bx,\bx)-\bk_n^\top(\bx) \mathbf{K}_n^{-1}\bk_n(\bx)+\bh^{\top}_n(\bx) \mathbf{G}_n^{-1} \bh_n(\bx)\right),
\end{align}
where $\hat{\boldsymbol{\beta}}_n = \mathbf{G}_n^{-1} \mathbf{P}_n^\top \mathbf{K}_n^{-1}\by_n$ is the maximum likelihood estimator (MLE) for $\boldsymbol{\beta}$, $\mathbf{G}_n = \mathbf{P}^\top_n \mathbf{K}_n^{-1} \mathbf{P}_n$ and $\bh_n(\bx) = \bp(\bx) -  \mathbf{P}_n^\top \mathbf{K}_n^{-1}\bk_n(\bx)$. These expressions can be equivalently viewed as the best linear unbiased predictor of $f(\bx)$ {under the presumed GP model} and its variance (see Section 3 of \citealp{santner2003design} for further details).

\subsection{Expected Improvement}

The EI method \citep{jones1998efficient} then makes use of the above closed-form equations to derive a closed-form acquisition function. Let $y_n^*=\min_{i=1}^n y_i$ be the current best objective value, and let $(y_n^*- f(\bx))_+ = \max\{y_n^*- f(\bx),0\}$ be the \textit{improvement} utility function. Given data $\cD_n$, the expected improvement acquisition function can be written as:
\begin{equation}
\EE_{f|\cD_n}(y_n^*- f(\bx))_+ = I_n(\bx)\Phi\left(\frac{I_n(\bx)}{\sigma s_n(\bx)}\right) + \sigma s_n(\bx) \phi\left(\frac{I_n(\bx)}{\sigma s_n(\bx)}\right).
\label{eqn:eidef}
\end{equation}
Here, $\phi(\cdot)$ and $\Phi(\cdot)$ denote the probability density function (p.d.f.) and cumulative density function (c.d.f.) of the standard normal distribution, respectively, and $I_n(\bx) = y_n^*-\hat{f}_n(\bx)$. For an unobserved point $\bx$, $\EE_{f|\cD_n}(y_n^*- f(\bx))_+$ can be interpreted as the expected improvement to the current best objective value over $\bx$, if the next query is at point $\bx$.

In order to compute the acquisition function in Equation \eqref{eqn:eidef}, we would need to know the process variance $\sigma^2$. In practice, however, this is typically unknown and needs to be estimated from data. A standard approach \citep{bull2011convergence} is to compute its MLE:
\begin{align}
    \hat{\sigma}^2_n  = \frac{1}{n}(\by_n -\mathbf{P}_n \hat{\boldsymbol{\beta}}_n)^\top \mathbf{K}_n^{-1} (\by_n -\mathbf{P}_n \hat{\boldsymbol{\beta}}_n),
\end{align}
then plug-in this into the acquisition function \eqref{eqn:eidef}. This gives the following plug-in estimate of the EI acquisition function:
\begin{align}\label{eqn:EI}
\textstyle \EI_n(\bx) = \underbrace{I_n(\bx)\Phi\left(\frac{I_n(\bx)}{\hat{\sigma}_n s_n(\bx)}\right)}_{\textbf{Exploitation}} + \underbrace{\hat{\sigma}_n s_n(\bx) \phi\left(\frac{I_n(\bx)}{\hat{\sigma}_n s_n(\bx)}\right)}_{\textbf{Exploration}}.
\end{align}
{In practice, the length-scale parameters $\boldsymbol{\theta}$ of the GP are typically estimated via MLE and plugged into the estimated EI criterion \eqref{eqn:EI} as well. }

With \eqref{eqn:EI} in hand, the next query point $\bx_{n+1}$ is obtained by maximizing the EI acquisition function $\textrm{EI}_n(\bx)$:
\begin{equation}
\bx_{n+1} \leftarrow \argmax_{\bx\in\Omega} \textstyle \EI_n(\bx).
\label{eq:eiopt}
\end{equation}
This maximization of \eqref{eqn:EI} implicitly captures the aforementioned \textit{exploration-exploitation trade-off}: {exploration} of the feasible region and {exploitation} near the current best solution. The maximization of the first term in \eqref{eqn:EI} encourages \textit{exploitation}, since larger values are assigned for points $\bx$ with smaller predicted values $\hat{f}_n(\bx)$. The maximization of the second term in \eqref{eqn:EI} encourages \textit{exploration}, since larger values for points $\bx$ indicate larger (estimated) predictive standard deviation $\hat{\sigma}_n s_n(\bx)$.

However, one drawback of EI is that it fails to capture the full uncertainty of model parameters within the acquisition function $\mathrm{EI}_n(\bx)$. This results in an \textit{over-exploitation} of the fitted GP model for optimization, and an \textit{under-exploration} of the black-box function over $\Omega$. This over-greediness has been noted in several works, e.g., \cite{bull2011convergence} and \cite{qin2017improving}. In particular, Theorem~3 of \cite{bull2011convergence} showed that, for a common class of correlation functions for $K$ (see Assumption \ref{ass:K} later), there always exists some smooth function $f$ within a function space $\mathcal{H}_{\boldsymbol{\theta}}(\Omega)$ (defined later in Section~\ref{sec:the}) such that EI fails to find a global optimum of $f$. This is stated formally below:

\begin{proposition}[Theorem 3, \citealp{bull2011convergence}]
\label{prop:ei}
Suppose Assumption \ref{ass:K} holds with $\nu<\infty$. Suppose initial points are sampled according to some probability measure $F$ over $\Omega$. Let  $(\bx_i)_{i=1}^\infty$ be the points generated by maximizing $\textup{EI}_n$ in \eqref{eqn:EI}{, with plug-in MLEs for GP length-scale parameters $\boldsymbol{\theta}$ satisfying Assumption \ref{ass:theta} (introduced later)}. Then, for any $\epsilon>0$, there exist some $f\in \cH_{\boldsymbol{\theta}}(\Omega)$ and some constant $\delta>0$ such that
\begin{equation*}
\PP_F\left( \lim_{n\rightarrow \infty}y_n^* - \min_{\bx\in\Omega}f(\bx)\geq \delta\right)>1-\epsilon.    
\end{equation*}
\end{proposition}

\noindent This proposition shows that, even for relatively simple objective functions $f$, the EI may fail to converge to a global minimum due to its over-exploitation of the fitted GP model. This is concerning not only from an asymptotic perspective, but also in practical problems with limited samples. As we show later, this overexploitation can cause the EI to get stuck on suboptimal solutions, resulting in poor empirical performance compared to the HEI.

\subsection{Existing solutions and limitations}

{As mentioned in the introduction, there are several notable methods in the literature which aim to tackle this over-exploitation. One approach which has proven effective is to \textit{integrate uncertainty on GP model parameters} within the EI criterion. An early work in this direction is \cite{snoek2012practical}, which made use of a fully Bayesian formulation of the EI with a full quantification of uncertainty in all GP model parameters. This integration of parameter uncertainty has subsequently been developed in recent works; see, e.g., \cite{chen2017sequential}. While effective, a key limitation with these methods is that this integration of uncertainty requires not only MCMC sampling of model parameters, but also the use of such samples within the optimization of the estimated EI criterion \eqref{eq:eiopt}. Particularly in higher dimensions, where many MCMC samples are required, this can greatly slow down and even worsen the optimization performance of subsequent points via \eqref{eq:eiopt}.}

{Another approach, as recommended in \cite{bull2011convergence}, is to encourage further exploration via an artificial inflation of the MLE for the variance parameter. In particular, in place of the MLE $\hat{\sigma}^2_n$, one instead uses the inflated estimator $n \hat{\sigma}^2_n$ within the estimated EI criterion \eqref{eqn:EI}. Along with an $\epsilon$-greedy modification, where points are added uniformly-at-random with probability $\epsilon$ at each iteration, \cite{bull2011convergence} proved this modified EI approach (which we call the $\epsilon$-EI) can indeed achieve convergence and a near-minimax convergence rate for global optimization, thus addressing the earlier limitation from Proposition \ref{prop:ei}. Despite its appealing theoretical properties, the $\epsilon$-EI may yield a suboptimal trade-off between exploration and exploitation, particularly with limited function evaluations. This may be due to the use of an artificially inflated variance parameter (which is ad-hoc from a modeling perspective), or the use of uniformly-random points within $\epsilon$-greedy sampling (which may be inefficient). This then results in suboptimal optimization performance given limited function evaluations, which we show later in numerical experiments.}

\section{Hierarchical Expected Improvement}\label{sec:model}

{To address the aforementioned limitations, we present a new Hierarchical EI (HEI) method, which integrates parameter uncertainty within a \textit{closed-form} acquisition function, thus allowing for efficient optimization of sequential queries. In employing a hierarchical Bayesian framework for expected improvement, the HEI can be shown to enjoy similar theoretical convergence guarantees as the $\epsilon$-EI in \cite{bull2011convergence} for global optimization, thus addressing the limitation from Proposition \ref{prop:ei} for the standard EI. In contrast to the $\epsilon$-EI, however, the HEI achieves this purely via a hierarchical Bayesian model, without a need for an ad-hoc variance inflation or $\epsilon$-greedy sampling. As such, the HEI can provide an improved trade-off between exploration and exploitation in practice, which then translates to better optimization performance given limited function evaluations, as we show later.}

A key ingredient for the HEI is a \textit{hierarchical} GP model on $f(\bx)$. Let us adopt the universal kriging model \eqref{eqn:uk}, but now with the following hierarchical priors assigned on two \cbl{independent} parameters $(\boldsymbol{\beta},\sigma^2)$:
\begin{align}\label{eqn:gp}
\textstyle [\boldsymbol{\beta}]&\propto \mathbf{1},\quad \quad [\sigma^2]\sim \textrm{IG}(a,b).
\end{align}

In words, the coefficients $\boldsymbol{\beta}$ are assigned a flat improper ({i.e.}, non-informative) prior over $\RR^q$, and the process variance $\sigma^2$ is assigned a conjugate inverse-Gamma prior with shape and scale parameters $a$ and $b$, respectively. The idea is to leverage this hierarchical structure on model parameters to account for estimation uncertainty, while preserving a closed-form criterion for efficient sequential sampling. {In what follows, we first introduce the HEI for \textit{fixed} GP length-scale parameters $\boldsymbol{\theta}$ for ease of exposition; Section \ref{sec:alg} then presents the full HEI procedure with \textit{estimated} length-scales, with corresponding theoretical analysis in Section \ref{sec:the}.}

The next lemma provides the posterior distribution of $f(\bx)$ under this hierarchical model:

\begin{lemma}\label{thm:post}
Assume the universal kriging model \eqref{eqn:uk} with hierarchical priors \eqref{eqn:gp} and $n > q$. Given data $\mathcal{D}_n$, we have
\begin{align}
\big[\sigma^2 \big|\cD_n  \big]\sim \mathrm{IG}\big(a_n,b_n \big)\quad \textrm{and}\quad \big[\boldsymbol{\beta} \big| \cD_n  \big]\sim T_q \big(2a_n,\hat{\boldsymbol{\beta}}_n , \tilde{\sigma}_n^2 \mathbf{G}_n^{-1}\big),
\end{align}
where $a_n = a+(n-q)/2$, $b_n = b+n\hat{\sigma}_n^2/2$, $\tilde{\sigma}_n^2 = b_n/a_n$, and $T_q(\nu,\boldsymbol{\mu},\boldsymbol\Sigma)$ is a $q$-dimensional non-standardized t-distribution with degrees of freedom $\nu$, location vector $\boldsymbol{\mu}$ and scale matrix $\boldsymbol\Sigma$. Furthermore, the posterior distribution of $f(\bx)$ is 
\begin{align}\label{eqn:post}
\big[ f(\bx)\big| \cD_n \big] \sim T_1\big(2a+n-q, \hat{f}_n(\bx), \tilde{\sigma}_n s_n(\bx) \big).
\end{align}
\end{lemma}
\noindent The proof of this lemma follows from Chapter 4.4 of \cite{santner2003design}. Lemma~\ref{thm:post} shows that under the universal kriging model~\eqref{eqn:uk} with hierarchical priors~\eqref{eqn:gp}, the posterior distribution of $f(\bx)$ is now a non-standarized $t$-distribution, with closed-form expressions for its location and scale parameters $\hat{f}_n(\bx)$ and $\tilde{\sigma}_n s_n(\bx)$. 

Comparing the predictive distributions in \eqref{eqn:model} and  \eqref{eqn:post}, there are several differences which highlight the increased predictive uncertainty from the hierarchical GP model. First, the new posterior \eqref{eqn:post} is now $t$-distributed, whereas the earlier posterior \eqref{eqn:model} is normally distributed, which suggests that the hierarchical model imposes heavier tails. 
Second, the scale term $\tilde{\sigma}^2_n$ in \eqref{eqn:post} can be decomposed as:
\begin{align}\label{eqn:bnan}
\tilde{\sigma}^2_n = (2b+n \hat{\sigma}_n^2)/(2a+(n-q)) > n/(2a+(n-q)) \cdot \hat{\sigma}_n^2.
\end{align}
When $a < q/2$ (which is satisfied via a weakly informative prior on $\sigma^2$), $\tilde{\sigma}^2_n$ is larger than the MLE $\hat{\sigma}_n^2$, which again increases predictive uncertainty.

Similar to the EI criterion \eqref{eqn:eidef}, we now define the HEI acquisition function as:
\begin{equation}
\textstyle\HEI_n(\bx) = \EE_{f|\cD_n}(y_n^*- f(\bx))_+,
\label{eqn:heidef}
\end{equation}
where the conditional expectation over $[f(\bx)|\cD_n]$ is under the hierarchical GP model. The proposition below gives a \textit{closed-form} expression for $\textstyle\HEI_n(\bx)$:
\begin{proposition}\label{thm:nei}
Assume the universal kriging model \eqref{eqn:uk} with hierarchical priors \eqref{eqn:gp} and $n>q$. Then:
\begin{align}\label{eqn:hei}
\textstyle\HEI_n(\bx)\hspace{-0.025in} =\hspace{-0.025in} \underbrace{I_n(\bx)\Phi_{\nu_n}\hspace{-0.025in} \left(\hspace{-0.025in}\frac{I_n(\bx)}{\tilde{\sigma}_ns_n(\bx)}\hspace{-0.025in}\right)}_{\mathbf{Exploitation}} \hspace{-0.025in}+\hspace{-0.025in} \underbrace{m_n\tilde{\sigma}_ns_n(\bx)\phi_{\nu_n-2}\hspace{-0.025in}\left(\hspace{-0.025in}\frac{I_n(\bx)}{m_n\tilde{\sigma}_ns_n(\bx)}\hspace{-0.025in}\right)}_{\mathbf{Exploration}} ,
\end{align}
where $m_n =\sqrt{\nu_n/(\nu_n-2)}$, $\nu_n = 2a_n$, and $\phi_{\nu_n}(x)$, $\Phi_{\nu_n}(x)$ denote the p.d.f. and c.d.f. of a Student's t-distribution with $\nu_n$ degrees of freedom, respectively. 
\end{proposition}
\noindent Proposition~\ref{thm:nei} shows that the HEI criterion preserves the desirable properties of original EI criterion \eqref{eqn:EI}: it has an easily-computable, closed-form expression, which allows for efficient optimization of the next query point. This HEI criterion also has an equally interpretable exploration-exploitation trade-off. Similar to the EI criterion, the first term encourages exploitation near the current best solution $\bx_n^*$, and the second term encourages exploration of regions with high predictive variance.

More importantly, the differences between the HEI \eqref{eqn:hei} and the EI \eqref{eqn:EI} acquisition functions highlight how our approach addresses the over-greediness issue. There are three notable differences. First, the HEI exploration term depends on the $t$-p.d.f. $\phi_{\nu_n-2}$, whereas the EI exploration term depends on the normal p.d.f. $\phi$. Since the former has heavier tails, the HEI exploration term is inflated, which encourages exploration. Second, the larger scale term $\tilde{\sigma}^2_n$ (see \eqref{eqn:bnan}) also inflates the HEI exploration term and encourages exploration. Third, the HEI contains an additional adjustment factor $\sqrt{\nu_n/(\nu_n-2)}$ in its exploration term. Since this factor is larger than 1, HEI again encourages exploration. This adjustment is most prominent for small sample sizes, since the factor $\sqrt{\nu_n/(\nu_n-2)}\rightarrow 1$ as sample size $n \rightarrow \infty$. All three differences \textit{correct} the over-exploitation of EI via a principled hierarchical Bayesian framework. We will show later that these modifications for the HEI address the aforementioned theoretical and empirical limitations of the EI.

Finally, we note that the Student EI, proposed by \cite{benassi2011robust}, can be viewed as a special case of the HEI criterion, with (i) a constant mean function $\mu(\bx) = \mu$, and (ii) ``fixed'' hyperparameters $a$ and $b$ (in that it does not scale with sample size $n$) for the inverse-Gamma prior in \eqref{eqn:gp}. We will show later that the HEI, by generalizing (i) and (ii), can yield improved theoretical and empirical performance over the SEI. For (i), instead of the \textit{stationary} GP model used in the SEI, the HEI instead considers a broader \textit{non-stationary} GP model with mean function $\mu(\bx) = \bp^\top(\bx) \boldsymbol{\beta}$, and factors in the uncertainty on coefficients $\boldsymbol{\beta}$ for optimization. This allows HEI to integrate uncertainty on GP nonstationarity to encourage more exploration in sequential sampling. For (ii), \cite{benassi2011robust} recommended a ``fixed'' hyperparameter setting for the SEI, where the hyperparameters $a$ and $b$ do not scale with sample size $n$. However, the following proposition shows that the SEI (under such a setting) can fail to find the global optimum, under mild regularity conditions (see Assumptions \ref{ass:K} and \ref{ass:theta} later).
\begin{proposition}
\label{prop:sei}
Suppose Assumptions \ref{ass:K} and \ref{ass:theta} hold with $\nu<\infty$. Suppose initial points are sampled according to some probability measure $F$ over $\Omega$. Given fixed hyperparameters $a$ and $b$, let $(\bx_i)_{i=1}^\infty$ be the points returned by the SEI procedure. Then, for any $\epsilon>0$, there exist some $f\in \cH_{\boldsymbol{\theta}}(\Omega)$ and some constant $\delta>0$ such that
\begin{equation*}
\PP_F\left( \lim_{n\rightarrow \infty}y_n^* - \min_{\bx\in\Omega}f(\bx)\geq \delta\right)>1-\epsilon.    
\end{equation*}
\end{proposition}

\noindent The proof is provided in Appendix A.2. This proposition shows that the SEI with ``fixed'' hyperparameters has the same limitation as the EI: it can fail to converge to a global minimum for relatively smooth objective functions $f$. We will show later in Section \ref{sec:the} that, under a more general prior specification which allows the hyperparameter $b$ to scale with sample size $n$ (more specifically, $b=\Theta(n)$), the HEI not only has the desired global convergence property for optimization, but also a near-minimax convergence rate.

\section{Methodology and Algorithm}\label{sec:hyper}

Using the HEI acquisition function \eqref{eqn:hei}, we now present a methodology for integrating this for effective black-box optimization. We first introduce hyperparameter estimation techniques which allow the HEI to mimic a fully Bayesian optimization procedure, then present an algorithmic framework for implementing the HEI. {For ease of exposition, Sections \ref{sec:hyp} and \ref{sec:ord} are presented with \textit{fixed} GP length-scale parameters $\boldsymbol{\theta}$; Section \ref{sec:alg} then presents the full algorithm with \textit{estimated} length-scales, with corresponding theoretical analysis in Section \ref{sec:the}.}



\subsection{Hyperparameter Specification}
\label{sec:hyp}
We present below several plausible specifications for the hyperparameters $(a,b)$ in the hierarchical prior $[\sigma^2] \sim \mathrm{IG}(a,b)$ in \eqref{eqn:gp}, and discuss when certain specifications may yield better optimization performance.

\vspace{0.05in}
\noindent{\bf (i) Weakly Informative.} Consider first a \textit{weakly informative} specification of the hyperparameters $(a,b)$, with $a=b=\epsilon$ for a small choice of $\epsilon$, {e.g.}, $\epsilon= 0.1$. This reflects weak information on the variance parameter $\sigma^2$, and provides regularization for parameter inference. The limiting case of $\epsilon\rightarrow 0$ yields the non-informative Jeffreys' prior for variance parameters.

While weakly informative (and non-informative) priors are widely used in Bayesian analysis \citep{gelman2006prior}, we have found that such priors can result in poor optimization performance for HEI (Section \ref{sec:exp} provides further details). One reason is that, for many black-box problems, only a small sample size can be afforded on the objective function $f$, since each evaluation is expensive. One can perhaps address this with a carefully elicited subjective prior, but such informative priors are typically not available when the objective $f$ is black-box. We present next two specifications which may offer improved optimization performance, both in theory (Section \ref{sec:the}) and in practice (Sections \ref{sec:exp} and \ref{sec:app}).

\vspace{0.05in}
\noindent{\bf (ii) Empirical Bayes.} Consider next an empirical Bayes (EB, \citealp{carlin2000bayes}) approach, which uses the observed data on $f$ to estimate the hyperparameters $(a,b)$. This is achieved by maximizing the following marginal likelihood for $(a,b)$: \begin{align}
p(\by_n;a, b) = \int \mathcal{L}(\boldsymbol{\beta},\sigma^2;\by_n) \pi(\boldsymbol{\beta}) \pi(\sigma^2;a,b) \; d \boldsymbol{\beta} d \sigma^2.
\end{align}
Here, $\mathcal{L}(\boldsymbol{\beta},\sigma^2;\by_n)$ is the likelihood function of the universal kriging model \eqref{eqn:uk} (see \citealp{santner2003design} for the full expression), and $\pi(\boldsymbol{\beta})$ and $\pi(\sigma^2;a,b)$ are the prior densities of $\boldsymbol{\beta}$ and $\sigma^2$ given hyperparameters $a$ and $b$. The model with estimated hyperparameters via EB provides a close approximation to a fully hierarchical Bayesian model \citep{carlin2000bayes}, where additional hyperpriors are assigned on $a$ and $b$. The latter can be viewed as a ``gold standard'' quantification of model uncertainty, but typically requires MCMC sampling, which can be more expensive than optimization. Here, an EB estimate of hyperparameters $(a,b)$ would allow the HEI to closely \textit{mimic} a fully Bayesian optimization procedure (the ``gold standard''), while avoiding expensive MCMC sampling via a \textit{closed-form} acquisition function.

Unfortunately, the proposition below shows that a direct application of EB for the HEI yields unbounded hyperparameter estimates:

\begin{proposition}\label{prop:eb}
The marginal likelihood for the universal kriging model \eqref{eqn:uk} with priors \eqref{eqn:gp} is given by:
\begin{align}\label{eqn:margin}
 p(\by_n;a, b) = \det(\mathbf{G}_n \mathbf{K}_n)^{-\frac{1}{2}} \frac{b^a}{\Gamma(a)} \frac{\Gamma(a+(n-q)/2)}{(b+w_n)^{a+\frac{n-q}{2}}},
\end{align}
where $w_n=(\by_n^\top \mathbf{K}_n^{-1} \by_n - \hat{\boldsymbol{\beta}}_n^\top \mathbf{G}_n \hat{\boldsymbol{\beta}}_n)/2$. 

Furthermore, the maximization problem:
\begin{equation}
\argmax_{a>0, b>0} \; p(\by_n;a,b)
\end{equation}
is unbounded for all values of $\by_n$. 
\end{proposition}
\noindent The proof of Proposition~\ref{prop:eb} is provided in Appendix A.3. 

To address this issue of unboundedness, one can instead use a modified EB approach, called the marginal maximum a posteriori estimator (MMAP, \citealp{doucet2002marginal}). The MMAP adds an additional level of hyperpriors $\pi(a,b)$ to the marginal likelihood maximization problem, yielding the modified formulation:
\begin{align}
   \cbl{\argmax_{a>0, b>0} \; \tilde{p}(\by_n;a, b) :=  \argmax_{a>0, b>0} \; p(\by_n;a, b) \pi(a,b).}
   \label{eqn:mmap}
\end{align}
The MMAP approach for hyperparameter specification has been used in a variety of problems, e.g., scalable training of large-scale Bayesian networks~\citep{JMLR:v14:liu13b}. The next proposition shows that the MMAP indeed yields finite solutions for a general class of hyperpriors on $(a,b)$:
\begin{proposition}
Assume the following independent hyperpriors on $(a,b)$:
\begin{equation}
[a] \sim \textup{Gamma}(\zeta,\iota), \quad \quad [b]\propto \mathbf{1}, 
\label{eq:hyperprior}
\end{equation}
where $\zeta$ and $\iota$ are the shape and scale parameters, respectively. Then the maximization of $\tilde{p}(\by_n;a, b)$ is always finite for $(a,b)\in\RR_+^2$.
\label{prop:finite}
\end{proposition}
\noindent The proof of Proposition~\ref{prop:finite} is provided in Appendix A.4.

{In practice, we recommend a weakly-informative specification of the hyperparameters $(\zeta,\iota)$ (i.e., with $\zeta = \iota$ set to be small), which appears to yield robust optimization performance for the HEI-MMAP. We note that the specification \eqref{eq:hyperprior} is simply one which works well in our implementation; given prior knowledge, a modeler has the flexibility of specifying an alternate prior which captures such information.} By mimicking a fully Bayesian optimization procedure, this MMAP approach can consistently outperform the weakly informative specification for HEI; we will show this later in numerical studies.

\vspace{0.05in}
\noindent{\bf (iii) Data-Size-Dependent (DSD).} Finally, we consider the so-called ``data-size-dependent'' (DSD) hyperparameter specification. This is motivated from the prior specification needed for global optimization convergence of the HEI, which we present and justify in the following section. The DSD specification requires the shape parameter $a$ to be constant, and the scale parameter $b$ to grow at the same order as the sample size $n$, {i.e.}, $b=\kappa n$ for some constant $\kappa > 0$. One appealing property of this specification is that it ensures the HEI converges to a global optimum $\bx^*$ (see Theorem \ref{thm:con} later).

For the DSD specification, we can similarly use the MMAP to estimate hyperparameters $(a,\kappa)$, to mimic a fully Bayesian EI procedure. Suppose data $\cD_{n_{\textrm{ini}}}$ are collected from $n_{\textrm{ini}}$ initial design points (more on this in Section \ref{sec:alg}). Then the hyperparameters $a$ and $\kappa$ can be estimated via the MMAP optimization:
\begin{align}
\textstyle (a^*,\kappa^*)  = \underset{a>0,\kappa>0}{\argmax} \left\{ p(\by_{n_{\textrm{ini}}};a, \kappa n_{\textrm{ini}}) \pi(a,\kappa) \right\},
\end{align}
where $\pi(a,\kappa)$ is the hyperprior density on $a$ and $\kappa$. One possible setting for $\pi(a,\kappa)$ is a Gamma hyperprior on $a$ and a non-informative hyperprior $[\kappa] \propto 1$ (independent of $a$). By Proposition 6, this specification again yields a finite optimization problem for MMAP. Using these estimated hyperparameters, subsequent points are then queried using HEI with $a = a^*$ and $b = \kappa^* n$, where $n$ is the current sample size. 

\vspace{-0.1in}
\subsection{Order Selection for Basis Functions}
\label{sec:ord}

In addition to hyperparameter estimation, the choice of basis functions in $\bp(\bx)$ and the order selection of such bases are also important for an effective implementation of the HEI. In our experiments, we take these bases to be complete polynomials up to a certain order $l$. Letting $\mathcal{M}^{(l)}$ denote the polynomial model with maximum order $l$, we have $\bp(\bx) = 1$ for model $\cM^{(0)}$ (a constant model), $\bp(\bx) = [1,x_1,\cdots, x_d]^\top$ for model $\cM^{(1)}$ (a linear model), $\bp(\bx) = [1,x_1,\cdots, x_d, x_1^2,\cdots, x_d^2,x_1x_2,\cdots,x_1x_d,x_2x_3,\cdots,x_{d-1}x_d]^\top$ for model $\cM^{(2)}$ (a second-order interaction model), etc. One can also make use of other basis functions (e.g., orthogonal polynomials; \citealp{xiu2010numerical}) depending on the problem at hand.

A careful selection of order $l$ is also important: an overly small estimate of $l$ results in over-exploitation of a poorly-fit model, whereas an overly large estimate results in variance inflation and over-exploration of the domain. We found that the standard Bayesian Information Criterion (BIC) \citep{schwarz1978estimating} provides good order selection performance for the HEI.  Given initial data $\cD_{n_{\textrm{ini}}}$, the BIC selects the model $\mathcal{M}^{(l^*)}$ with order:
\begin{align}\label{eqn:bic}
l^*=    \argmin_{l\in \mathbb{N}} \big\{-2 \log \mathcal{L}(\cM^{(l)}) + q_l \log(n_{\textrm{ini}})\big\}.
\end{align}
Here, $\mathcal{L}(\mathcal{M}^{(l)})$ denotes the likelihood of model $\cM^{(l)}$ (this likelihood expression can be found in \citealp{santner2003design}), and $q_l$ denotes the number of basis functions in model $\cM^{(l)}$. With this optimal order selected, subsequent samples are then obtained using HEI with mean function following this polynomial order.

\vspace{-0.1in}

\subsection{Algorithm Statement}\label{sec:alg}

\RestyleAlgo{boxruled}
\LinesNumbered
\begin{algorithm}[t!]
	\caption{Hierarchical Expected Improvement for Bayesian Optimization}\label{alg:BO}
	    \vspace{0.2cm}
		\textbf{Initialization}
		\vspace{-0.3cm}
		\begin{itemize}
		\setlength{\itemsep}{-5pt}
		\item Generate $n_{\textrm{ini}}$ space-filling design points $\{\bx_1, \cdots, \bx_{n_{\textrm{ini}}}\}$ on $\Omega$.
		\item Evaluate function points $y_i = f(\bx_i)$, yielding the initial dataset $\cD_{n_{\textrm{ini}}} = \{(\bx_i,y_i)\}_{i=1}^{n_{\textrm{ini}}}$.
		\end{itemize}
		\vspace{-0.2cm}
		\textbf{Model selection}
		\vspace{-0.3cm}
		\begin{itemize}
		\setlength{\itemsep}{-5pt}
		    \item Select model order via BIC using \eqref{eqn:bic}.
		    \item Estimate hyperparameters $(a,b)$ via MMAP using \eqref{eqn:mmap}.
		\end{itemize}
		
	    \vspace{-0.2cm}
		\textbf{Optimization}\\
		\For{\textrm{ $n$  $\gets n_{\mathrm{ini}}$ to $n_{\mathrm{tot}}-1$}}{
		\vspace{-0.3cm}
		\begin{itemize}
		\setlength{\itemsep}{-5pt}
		    \item Given $\cD_n$, estimate length-scale parameters $\boldsymbol{\theta}$ via MAP and compute $\HEI_n(\bx)$.
		\item Obtain the next evaluation point $\bx_{n+1}$ by maximizing $\mathrm{HEI}_n(\bx)$: 
		\begin{equation}
		    \bx_{n+1} \leftarrow \argmax_{\bx\in\Omega} \mathrm{HEI}_n(\bx).
		    \label{eq:algeqn}
		\end{equation}
		\item Evaluate $y_{n+1} = f(\bx_{n+1})$, and update data $\cD_{n+1} = \cD_{n}\cup \{(\bx_{n+1},y_{n+1})\}.$
		\end{itemize}
}
		\textbf{Return: } The best observed solution $\bx_{i^*},$ where $i^* = \argmin_{i = 1}^{n_{\textrm{tot}}}f(\bx_i)$.
\end{algorithm}

Algorithm~\ref{alg:BO} summarizes the above steps for HEI. First, initial data on the black-box function $f$ are collected on a ``space-filling'' design, which provides good coverage of the feasible space $\Omega$. For the unit hypercube $\Omega = [0,1]^d$, we have found that the maximin Latin hypercube design (MmLHD, \citealp{morris1995exploratory}) works quite well in practice. For non-hypercube domains, more elaborate design methods on non-hypercube regions (e.g., \citealp{lekivetz2015fast, mak2018minimax, joseph2019designing}) can be used. The number of initial points is set as $n_{\rm ini} = 10d$, as recommended in \cite{loeppky2009choosing}. Using this initial data, the model order for the hierarchical GP is selected using \eqref{eqn:bic}. The hyperparameters $a$ and $b$ are also estimated from data (if necessary) using the methods described in Section \ref{sec:hyp}. Next, the following two steps are repeated until the sample size budget $n_{\rm tot}$ is exhausted: (i) the GP length-scale parameters $\boldsymbol{\theta}$ are fitted via maximum a posteriori (MAP) estimation\footnote{In numerical experiments, we use an independent uniform prior $\theta_l \sim^{i.i.d.} U[0,100]$ for this MAP estimate.} using the observed data points, (ii) a new sample $f(\bx)$ is collected at the point $\bx$ which maximizes the HEI criterion \eqref{eqn:hei}.

\section{Convergence Analysis}\label{sec:the}

We present next the global optimization convergence result for the HEI, then provide a near-minimax optimal convergence rate for the proposed method. In what follows, we will assume that the domain $\Omega$ is convex and compact.

Let us first adopt the following shift-invariant form for the kernel $K$:
\begin{align}
\label{eq:ker}
 K_{\boldsymbol{\theta}}(\bx,\bz) := C\left(\frac{x_1-z_1}{\theta_1},\ldots, \frac{x_d-z_d}{\theta_d}\right),
\end{align}
where $C$ is a stationary correlation function with $C(\mathbf{0}) = 1$ and length-scale parameters $\boldsymbol{\theta} = (\theta_1, \cdots, \theta_d)$. From this, we can then define a function space -- the reproducing kernel Hilbert space (RKHS, \citealp{wendland2004scattered}) -- for the objective function $f$. Given kernel $K_{\boldsymbol{\theta}}$ (which is symmetric and positive definite), define the linear space
\begin{align}
   \cF_{\boldsymbol{\theta}}(\Omega) = \left\{\sum_{i=1}^N \alpha_i K_{\boldsymbol{\theta}}(\cdot,\bx_i): N\in\NN_+, \bx_i\in \Omega, \alpha_i\in\RR\right\},
\end{align}
and equip this space with the bilinear form
\begin{align}
  \left\langle \sum_{i=1}^N \alpha_i K_{\boldsymbol{\theta}}(\cdot,\bx_i),\sum_{j=1}^M \gamma_j K_{\boldsymbol{\theta}}(\cdot,\by_j) \right\rangle_{K_{\boldsymbol{\theta}}}: =\sum_{i=1}^N\sum_{j=1}^M \alpha_i \gamma_j K_{\boldsymbol{\theta}}(\bx_i,\by_j).
\end{align}
The RKHS $\cH_{\boldsymbol{\theta}}(\Omega)$ of kernel $K_{\boldsymbol{\theta}}$ is defined as the closure of $\cF_{\boldsymbol{\theta}}(\Omega)$ under
$\langle\cdot,\cdot\rangle_{K_{\boldsymbol{\theta}}}$, with its
inner product $\langle \cdot,\cdot \rangle_{\cH_{\boldsymbol{\theta}}}$ induced by 
$\langle\cdot,\cdot\rangle_{K_{\boldsymbol{\theta}}}$ \citep{wendland2004scattered}.

Next, we make the following two regularity assumptions. The first is a smoothness assumption on the correlation function $C$: 
\begin{assumption} \label{ass:K}
$C$ is continuous, integrable, and satisfies:
\begin{equation*}
 |C(\bx) - Q_r(\bx) | = \mathcal{O}\left(\norm{\bx}_2^{2\nu}(-\log \norm{\bx}_2)^{2\alpha}\right) \quad \text{as} \quad \norm{\bx}_2\rightarrow 0,
\end{equation*}
for some constants $\nu>0$ and $\alpha\geq 0$. Here, $r = \lfloor 2\nu \rfloor$ and $Q_r(\bx)$ is the $r$-th order Taylor approximation of $C(\bx)$. Furthermore, its Fourier transform 
$$\widehat{C}(\boldsymbol{\xi}):=\int_{\bx}e^{-2\pi i \langle \boldsymbol{\xi},\bx \rangle}C(\bx)d\bx$$ is isotropic, radially non-increasing and satisfies either: as $\norm{\bx}_2\rightarrow \infty$
\begin{center}
$\widehat{C}(\bx) = \Theta \left(\norm{\bx}_2^{-2\nu-d}\right)$ \quad \text{or} \quad $\widehat{C}(\bx) = \mathcal{O}\left(\norm{\bx}_2^{-2\lambda-d}\right)$ \quad for any $\lambda > 0$.
\end{center}
\end{assumption}
\noindent Note that in Assumption~\ref{ass:K}, since we assume $C$ is continuous and integrable, its Fourier transform $\hat{C}$ must exists. A widely-used correlation function which satisfies this assumption is the Mat\'ern correlation function (see \citealp{cressie1991statistics} and \citealp{bull2011convergence}). 
\cbl{The second assumption is a bounded assumption on the prior $\pi(\boldsymbol{\theta})$ for the GP length-scale parameters $\boldsymbol{\theta}$.}
\begin{assumption} \label{ass:theta}
We assume that the prior $\pi(\boldsymbol{\theta})$ is bounded and bounded away from $\mathbf{0}$. Thus, given data $\mathcal{D}_n$, if we let $\tilde{\boldsymbol{\theta}}_n$ be the MAP of $\boldsymbol{\theta}$ under prior $\pi(\boldsymbol{\theta})$, it follows that for any $n>q$:
\begin{align}\label{eqn:bound}
 \boldsymbol{\theta}^L\leq \tilde{\boldsymbol{\theta}}_n \leq \boldsymbol{\theta}^U \quad \textrm{for some constants $\boldsymbol{\theta}^L, \boldsymbol{\theta}^U\in\RR^d_+$}.
\end{align}
\end{assumption}
\noindent \cbl{Note that the above equation indicates component-wise inequalities.} This is \cbl{similar to} Definition 2 in \cite{bull2011convergence}.

Under these two regularity assumptions, we can then prove the global optimization convergence of the HEI method (more specifically, for HEI-DSD).
\begin{theorem}
\label{thm:con}
Suppose Assumptions~\ref{ass:K} and~\ref{ass:theta} hold. Further suppose the hyperparameter $a$ is a constant (in $n$) and $b = \Theta(n)$, with basis functions $p_i(x)\in\mathcal{H}_{\theta^U}(\Omega)$. Let $(\bx_i)_{i=1}^\infty$ be the points generated by maximizing $\textup{HEI}_n$ in \eqref{eqn:hei}, with iterative plug-in MAP estimates $\tilde{\boldsymbol{\theta}}_n$. Then, for any $f \in \cH_{\boldsymbol{\theta}^U}(\Omega)$ and any choice of initial points $\{\bx_i\}_{i=1}^{n_{\rm ini}}$, we have:
\begin{align}
 y_n^* - \min_{\bx\in\Omega} f(\bx) = \left\{
\begin{array}{ll}
\cO(n^{-\nu/d}(\log n)^\alpha),& \nu\leq 1,\\
\cO(n^{-1/d}),& \nu> 1.\\
\end{array}
\right.
\end{align}
\end{theorem}

\noindent The proof of this theorem is given in Appendix A.5, \cbl{where its dependence on other parameters (e.g., $\boldsymbol{\theta}^L$ and $\boldsymbol{\theta}^U$) are made explicit}. The key idea is to upper bound the prediction gap $f(\bx)-\hat{f}_n(\bx)$ by the posterior variance term $s^2_n(\bx)$~in \eqref{eqn:post_var}, which is a generalization of the power function used in the function approximation literature (see, e.g., Theorem 11.4 of \citealp{wendland2004scattered}). We then show that the hyperparameter assumption $b = \Theta(n)$ prevents the estimator $\tilde{\sigma}_n^2$ from collapsing to $0$, and allows us to apply approximation bounds on $s^2_n(\bx)$ to obtain the desired global convergence result. This proof is inspired by Theorem 4 of \cite{bull2011convergence}.

Theorem~\ref{thm:con} shows that, for all objective functions $f$ in the RKHS $\cH_{\boldsymbol{\theta}^U}(\Omega)$, the HEI indeed has the desired convergence property for global optimization. This addresses the lack of convergence for the EI from Proposition \ref{prop:ei} and the SEI from Proposition~\ref{prop:sei}. It is worth nothing that, when $C$ is the Mat\'ern correlation with smoothness parameter $\nu$ \citep{cressie1991statistics}, the RKHS $\cH_{\boldsymbol{\theta}}(\Omega)$ consists of functions $f$ with continuous derivatives of order $\nu' < \nu$ \citep{santner2003design}. Hence, for the Mat\'ern correlation, the HEI achieves a global convergence rate of $\cO(n^{-1/d})$ for objective functions $f\in \cH_{\boldsymbol{\theta}}(\Omega)$ with $\nu > 1$, and $\cO(n^{-\nu/d}(\log n)^\alpha)$ for objective functions $f\in \cH_{\boldsymbol{\theta}^U}(\Omega)$ with $\nu \leq 1$.

At first glance, the prior specification in Theorem \ref{thm:con} may appear slightly peculiar, since the hyperparameter $b=\Theta(n)$ depends on the sample size $n$. However, such \textit{data-size-dependent} priors have been studied extensively in the context of high-dimensional Bayesian linear regression, particularly in its connection to optimal minimax estimation (see, e.g., \citealp{castillo2015bayesian}). The data-size-dependent prior in Theorem \ref{thm:con} can be interpreted in a similar way: the hyperparameter condition $b=\Theta(n)$ is sufficient in encouraging exploration in the sequential sampling points, so that HEI converges to a global optimum for all $f$ in the RKHS $\cH_{\boldsymbol{\theta}^U}(\Omega)$.

\cbl{If an additional $\gamma$-stability condition \citep{wynne2020convergence} holds, we can further show that the HEI achieves the minimax convergence rate for Bayesian optimization (we provide discussion on this minimax rate at the end of the section). This condition is stated below:
\begin{condition}
\label{ass:stb}
Let $(\bx_i)_{i=1}^\infty$ be the sequence of points generated by the HEI. We assume that
\begin{align}\label{eq:cons}
    s_n(\bx_{n+1})\geq \gamma \norm{s_n(\bx)}_{\infty} \quad \text{ for all } n = 1, 2, \cdots,
\end{align}
for some constant $\gamma \in (0,1]$.
\end{condition}
}
\noindent In words, this requires that every sequential point $\bx_{n+1}$ has a posterior standard deviation term $s_n(\bx_{n+1})$ (from \eqref{eqn:post_var}) at least as large as $\gamma \norm{s_n(\bx)}_{\infty}$, where $\norm{s_n(\bx)}_{\infty}$ is the maximum posterior standard deviation term over domain $\Omega$. If this condition holds, we can then show a quicker convergence rate for global optimization:

\begin{theorem}\label{thm:stab}
Suppose Assumptions~\ref{ass:K} and \ref{ass:theta} hold, along with Condition \ref{ass:stb}. Further suppose the hyperparameter $a$ is a constant (in $n$) and $b = \Theta(n)$, with basis functions $p_i(x)\in\mathcal{H}_{\theta^U}(\Omega)$. Let $(\bx_i)_{i=1}^\infty$ be the points generated by maximizing $\textup{HEI}_n$ in \eqref{eqn:hei} with iterative plug-in MAP estimates $\tilde{\boldsymbol{\theta}}_n$. Then, for any $f \in \cH_{\boldsymbol{\theta}^U}(\Omega)$ and any initial points, we have:
\begin{align}
 y_n^* - \min_{\bx\in\Omega} f(\bx) = O(n^{-\nu / d }).
\end{align}

\end{theorem}
\noindent Further details on this theorem are provided in Appendix A.6. \begin{figure}[!tbh]
    \centering
    \subfigure[Log-ratio $\log\frac{s_n(\bx_{n+1})}{\norm{s_n(\bx)}_{\infty}}$ for HEI-DSD.]{
\includegraphics[width = 0.45\textwidth]{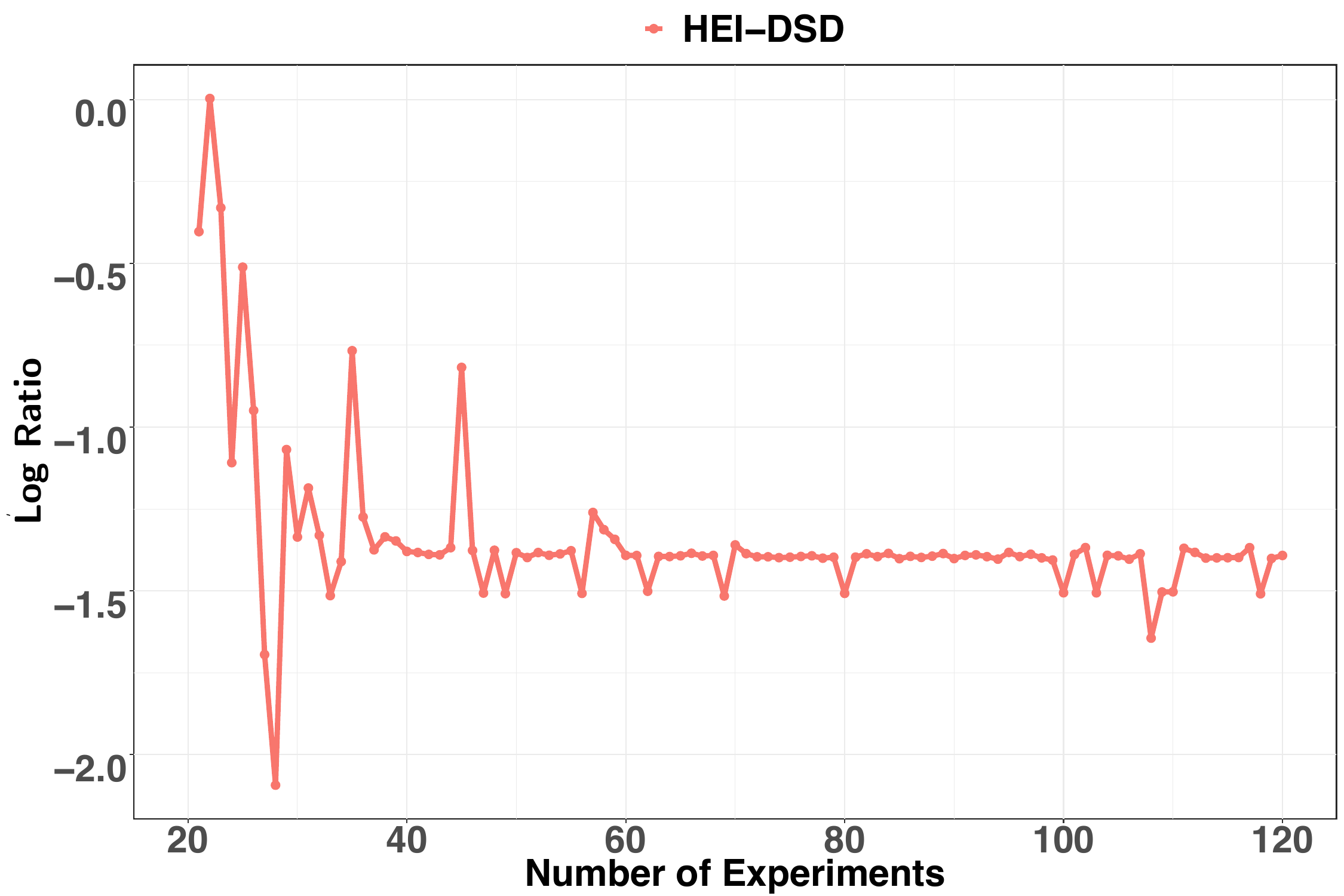}\label{fig:HEI-DSD-branin}
}
 \subfigure[Posterior standard deviation $\tilde{\sigma}_n$ for HEI methods.]{
\includegraphics[width = 0.45\textwidth]{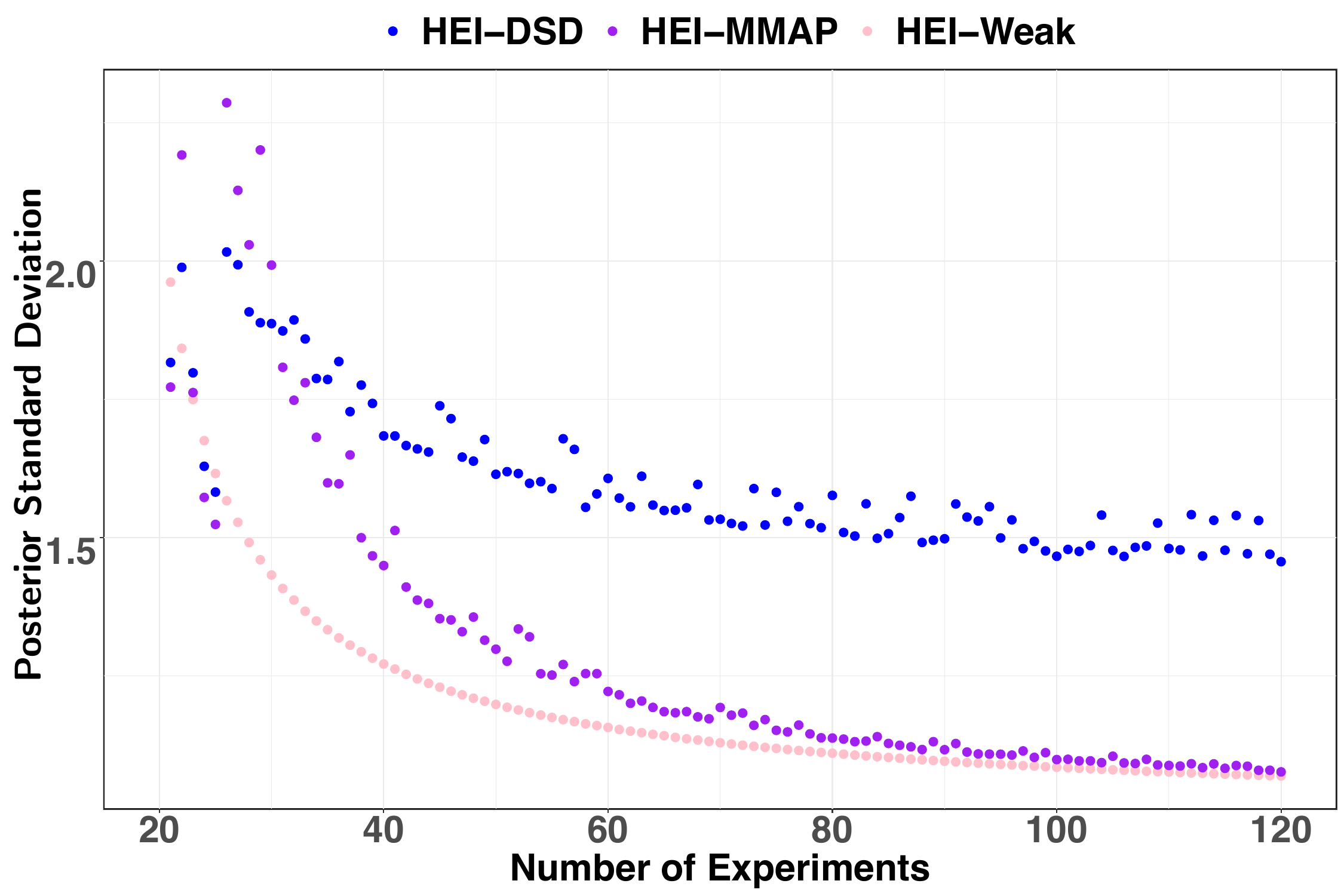}\label{fig:HEI-post-variance}
}
    \caption{{Visualizing select statistics for HEI methods on the Branin numerical experiment.}} 
    \label{fig:ratio}
\end{figure}

{\cbl{Condition \ref{ass:stb} is unfortunately quite difficult to guarantee theoretically, but appears to be satisfied empirically for the recommended HEI-DSD and HEI-MMAP methods in all of our later numerical experiments.} Figure~\ref{fig:HEI-DSD-branin} shows the log-ratio $\log \{{s_n(\bx_{n+1})}/{\norm{s_n(\bx)}_{\infty}}\}$ for HEI-DSD using the Branin function (in our later simulation study). As can be seen, the log ratio appears to be bounded away from zero as $n$ increases, which suggests that the HEI with data-size-dependent hyperparameter specification indeed satisfies Condition \ref{ass:stb} for a sufficiently small $\gamma  > 0$. This is not surprising, since while a sensible optimization algorithm would place more points around the optimum, the hierarchical nature of the HEI encourages further exploration, thus ensuring the sequential points explore the domain from time to time such that no point has overly high predictive variance. The same $\gamma$-stability assumption was used in \cite{wynne2020convergence} for proving convergence of existing BO methods.}

{It is crucial to note that, while these rates provides a reassuring check for HEI convergence, \textit{such asymptotic analysis does not tell the full story on the effectiveness of a Bayesian optimization method}. \cite{bull2011convergence} proved that, of all optimization strategies for minimizing $f \in \mathcal{H}_{\theta}(\Omega)$ under Assumption \ref{ass:K}, the minimax rate for the optimization gap $y_n^* - \min_{\bx\in\Omega} f(\bx)$ is $\mathcal{O}(n^{-\nu/d})$, i.e., there does not exist an optimization strategy with a quicker asymptotic rate. The HEI rate in Theorem \ref{thm:stab}, in this sense, is precisely the minimax rate. However, \cite{bull2011convergence} also showed that the simple (non-adaptive) strategy of optimization via a quasi-uniform sequence (see, e.g., \citealp{niederreiter1992random}) can also achieve this minimax rate! Such a strategy, however, typically performs terribly in practice and is not competitive with existing BO methods \citep{bull2011convergence}, since it is non-adaptive to observed function evaluations. This shows that \textit{such asymptotic analysis, while providing a reassuring check, cannot be used as a sole metric for gauging the practical effectiveness of different methods, particularly given the current setting of limited sample sizes}.}

{We further note that analogous rates to Theorems \ref{thm:con} and \ref{thm:stab} have been also shown for various Bayesian optimization methods in the literature. In particular, similar rates to Theorem \ref{thm:con} were proved for the $\epsilon$-EI \citep{bull2011convergence}, and our results leverage an adaptation of their proof techniques to show convergence for the HEI. Similarly, the $\gamma$-stability assumption (Condition \ref{ass:stb}) was used in \cite{wynne2020convergence} to establish similar optimization rates as Theorem \ref{thm:stab} for certain Bayesian optimization methods. The novelty for the HEI is thus not in terms of improved asymptotic rates over existing methods; such rates primarily serve to provide theoretical footing for our approach. Instead, the key contribution of the HEI is methodological: the proposed hierarchical framework provides a principled exploitation-exploration trade-off via a closed-form acquisition function, which as we show later, allows for improved optimization performance with limited samples.}

\section{Numerical Experiments}\label{sec:exp}

We now investigate the numerical performance of HEI in comparison to existing BO methods, for a suite of test optimization functions. We consider the following five test functions, taken from~\cite{simulationlib}: 

\setlength{\leftmargini}{10pt}
\begin{itemize}
\item \noindent{\bf Branin} (2-dimensional function on domain $\Omega =[0,1]^2$): 
$$\textstyle f(\bx) =(x_2-5.1/(4\pi^2)\cdot x_1^2+5/\pi \cdot x_1-6)^2+10(1-1/(8\pi))\cos(x_1)+10,$$
\item \noindent \textbf{Three-Hump Camel} (2-dimensional function on domain $\Omega =[-2,2]^2$):
$$\textstyle f(\bx) = 2x_1^2-1.05x_1^4+x_1^6/6+x_1x_2+x_2^2,$$
\item \noindent \textbf{Six-Hump Camel} (2-dimensional function on domain $\Omega =[-2,2]^2$):
$$\textstyle f(\bx) = (4-2.1x_1^2+x_1^4/3)x_1^2+x_1x_2+(-4+4x_2^2)x_2^2,$$
\item \noindent{\bf Levy Function} (6-dimensional function on domain $\Omega =[-10,10]^6$):
\vspace{-0.05in}
$$ f(\bx) = \sin^2(\pi \omega_1)+\sum_{i=1}^5(\omega_i-1)^2[1+10\sin^2(\pi\omega_i+1)]+(\omega_6-1)^2[1+\sin^2(2\pi\omega_6)],$$
where $\omega_i = 1+(x_i-1)/4$ for $i=1,\cdots,6$,
\item \noindent{\bf Ackley Function} (10-dimensional function on domain $\Omega =[-5,5]^{10}$):
\vspace{-0.05in}
$$ f(\bx) = -20 \exp\left(-\frac{0.2}{\sqrt{10}}\norm{\bx}_2\right)-\exp\left(\frac{1}{10}\sum_{i=1}^{10}\cos(2\pi x_i)\right)+20+\exp(1).$$
\end{itemize}

The simulation set-up is as follows. We compare the proposed HEI method under different hyperparameter specifications (HEI-Weak, HEI-MMAP, HEI-DSD), with the EI method under ordinary kriging (EI-OK) and universal kriging (EI-UK), the Student EI (SEI) method with fixed hyperparameters $(0.2,12)$ as recommended in \cite{benassi2011robust}, the UCB approach under ordinary kriging (UCB-OK, \citealp{Srinivas:2010:GPO:3104322.3104451}) with default exploration parameter $2.96$, the $\epsilon$-greedy EI approach \citep{bull2011convergence} under ordinary kriging ($\epsilon$-EI-OK) and universal kriging ($\epsilon$-EI-UK) with $\epsilon = 0.1$ as suggested in~\cite{sutton2018reinforcement}, and the $\gamma$-stabilized EI method \citep{wynne2020convergence} under universal kriging (Stab-EI-UK) with $\gamma = \min(0.1d, 0.8)$\footnote{Stab-EI-UK requires the next query point $x_{n+1}$ to satisfy $s_n(x_{n+1})\geq \gamma \norm{s_n(x)}_{\infty}$. In our implementation, we randomly sample $10^{d+2}$ points to find $\norm{s_n(x)}_{\infty}$ and set $\gamma =\min(0.1d, 0.8)$.}. For HEI-Weak, the hyperparameters $(a,b)$ are set as $a = b =0.1$; for HEI-MMAP and HEI-DSD, the hyperparameters $(\zeta,\iota)$ are set as $\zeta = \iota = 2$. All methods use the Mat\'ern correlation with smoothness parameter $2.5$, and are run for a total of $T=120$ function evaluations. Here, the kriging model is fitted using the \textsc{R} package \texttt{kergp} \citep{kergp}. {As mentioned in Section \ref{sec:alg}, all methods are initialized using maximin Latin hypercube designs \citep{morris1995exploratory}.} Simulation results are averaged over 20 replications except the Ackley function, due to the heavy computation burden for fitting the high-dimensional GP model.
\begin{figure}[!tbp]
\begin{center} 
\subfigure[Branin (2-d)]{
\includegraphics[width = .4\textwidth]{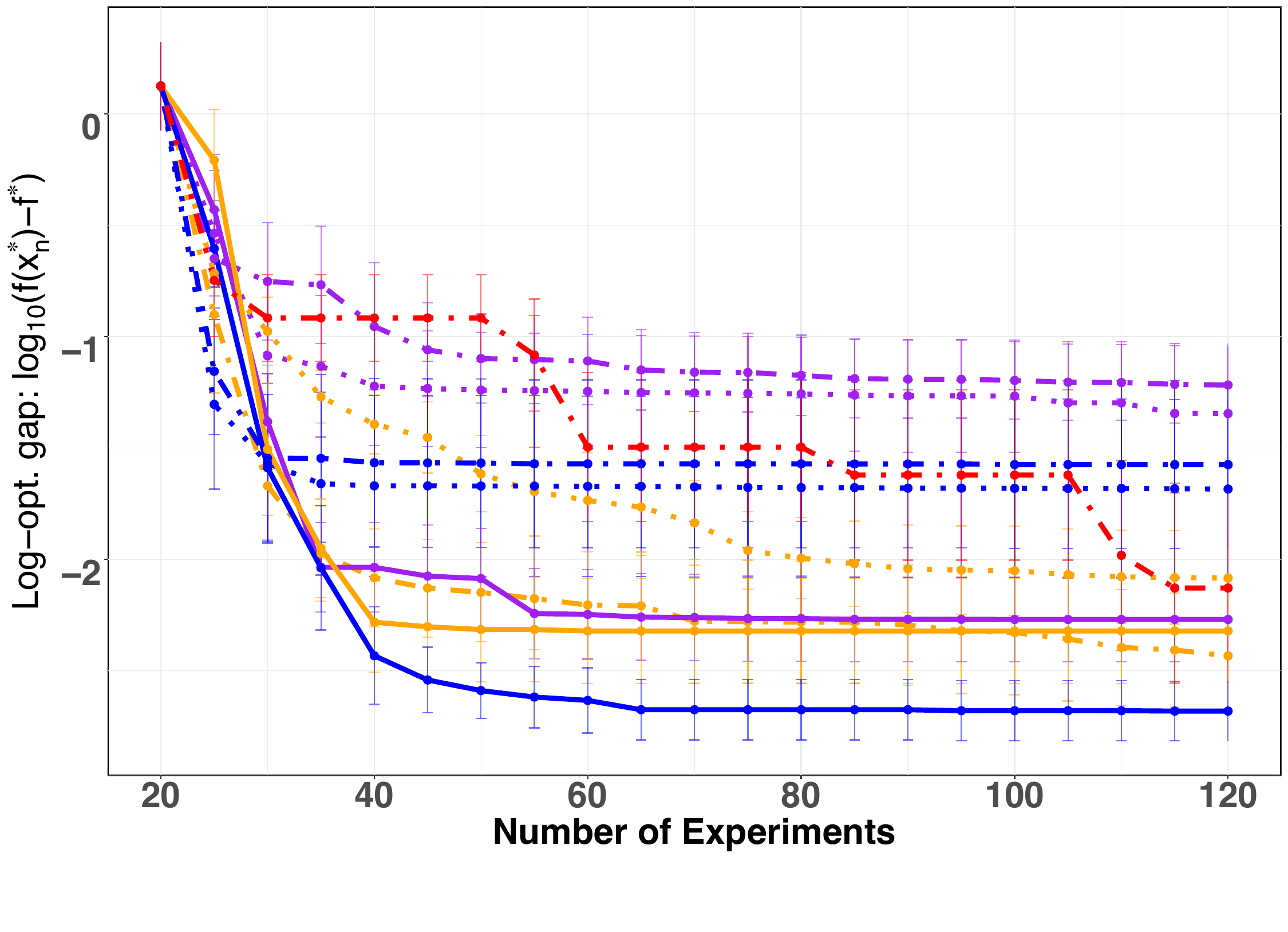}\label{fig:best}
}
\subfigure[A visualization of sampled points]{
\includegraphics[width = .45\textwidth]{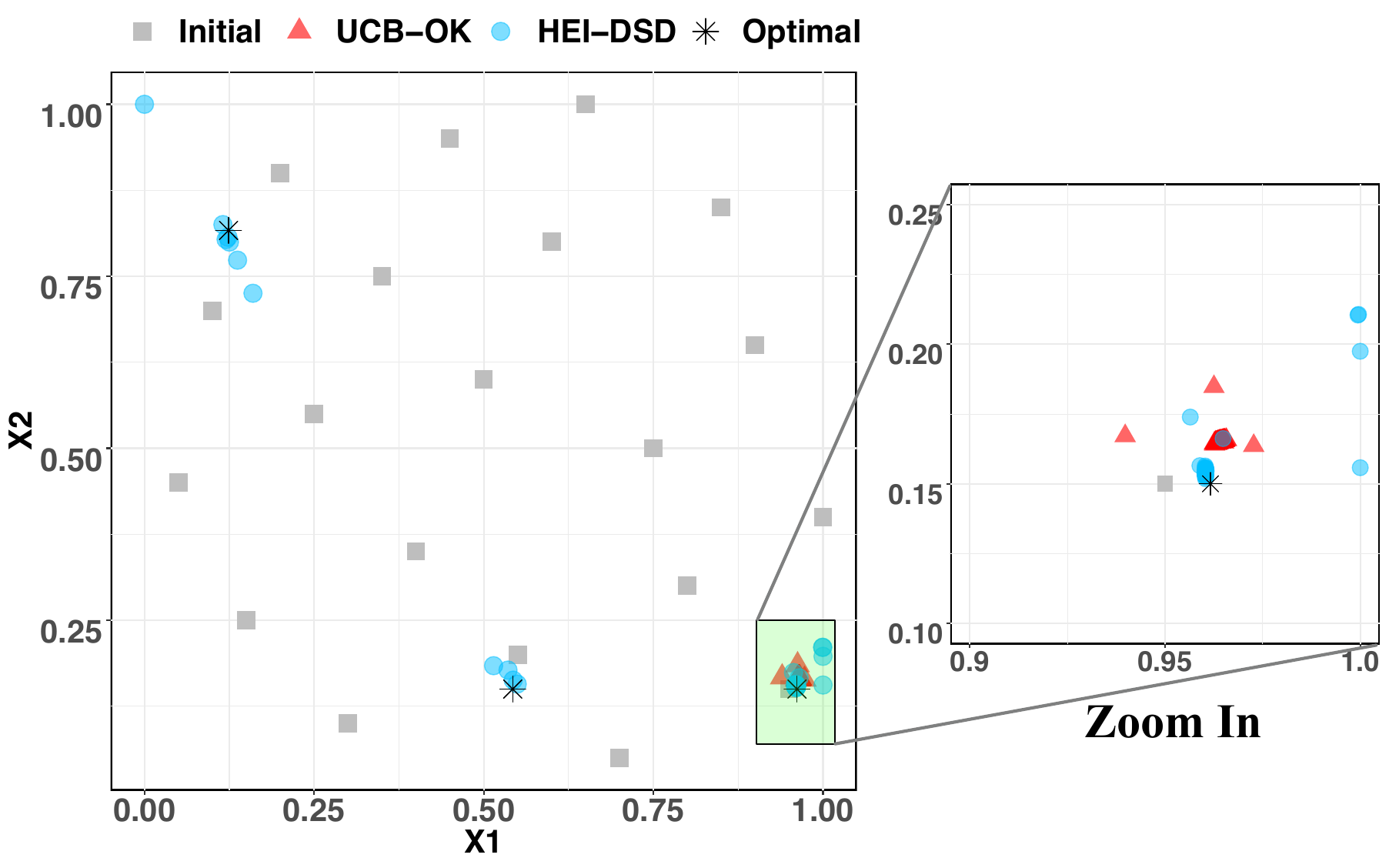}\label{fig:visual}
}\\
\subfigure[Three-Hump Camel (2-d)]
{\includegraphics[width = .4\textwidth]{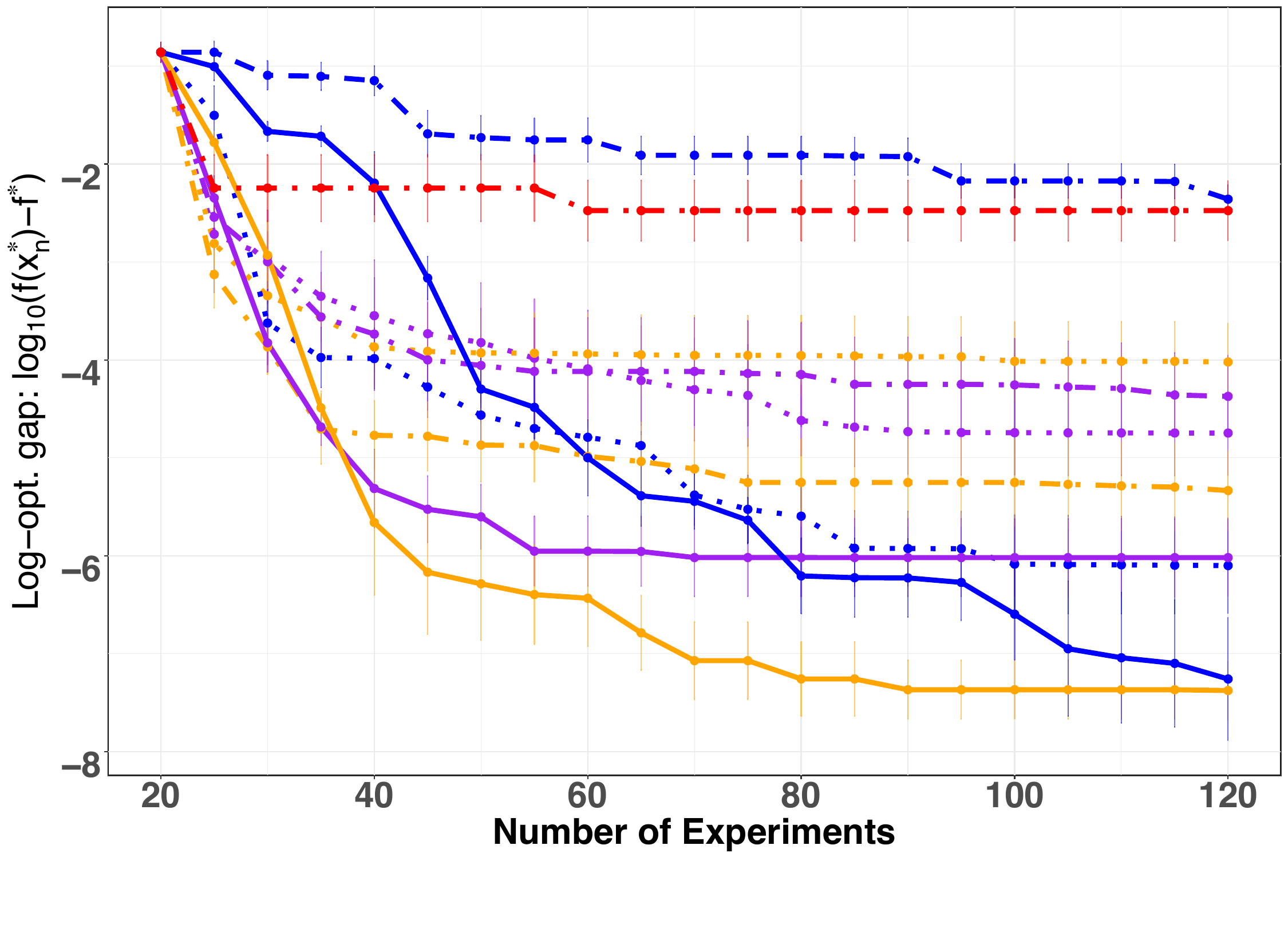}\label{fig:hump3:best}}\hspace{0.44in}
\subfigure[Six-Hump Camel (2-d)]
{\includegraphics[width = .4\textwidth]{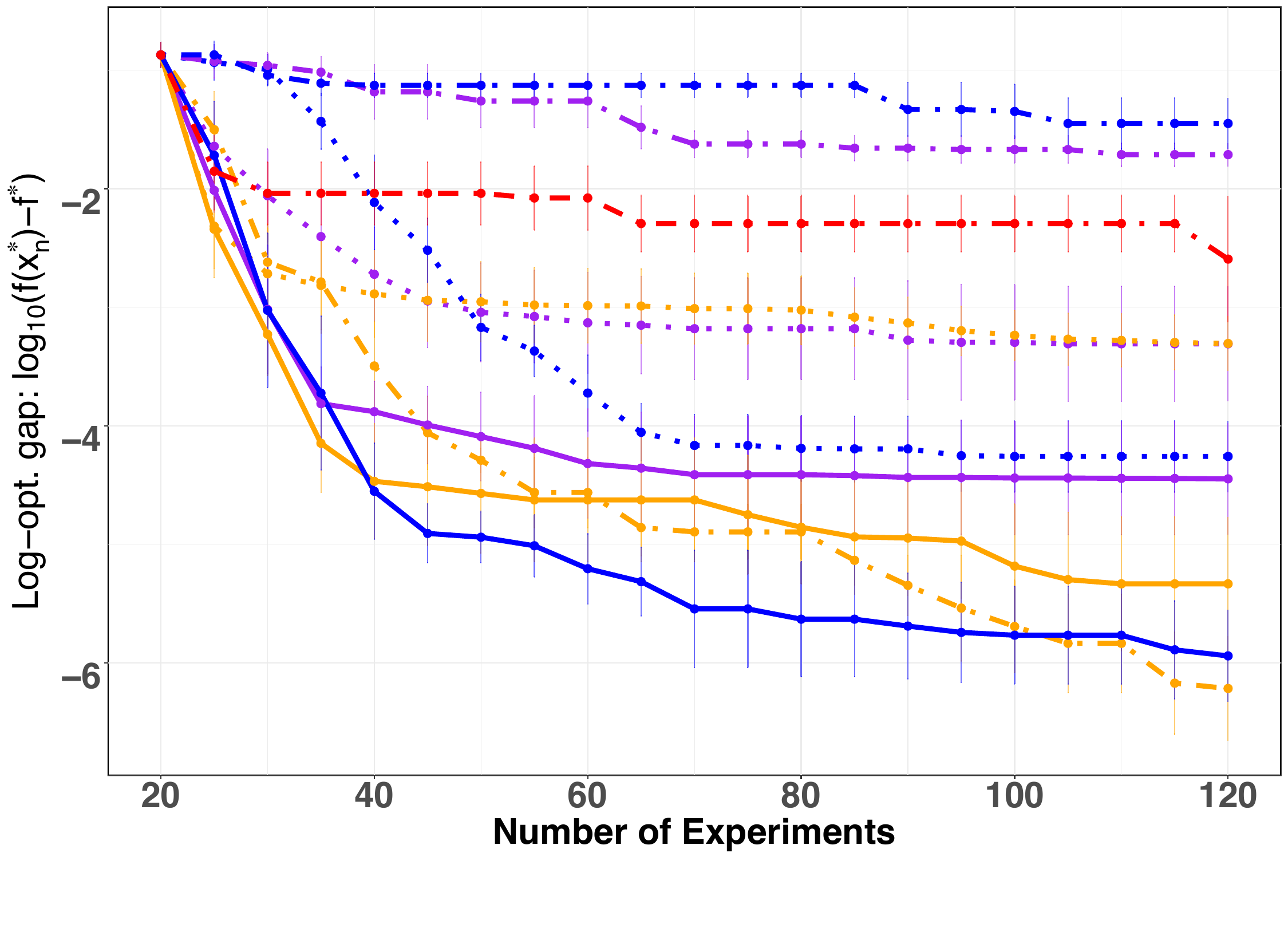}\label{fig:hump6:best}}\\
\subfigure[Levy (6-d)]
{\includegraphics[width = .4\textwidth]{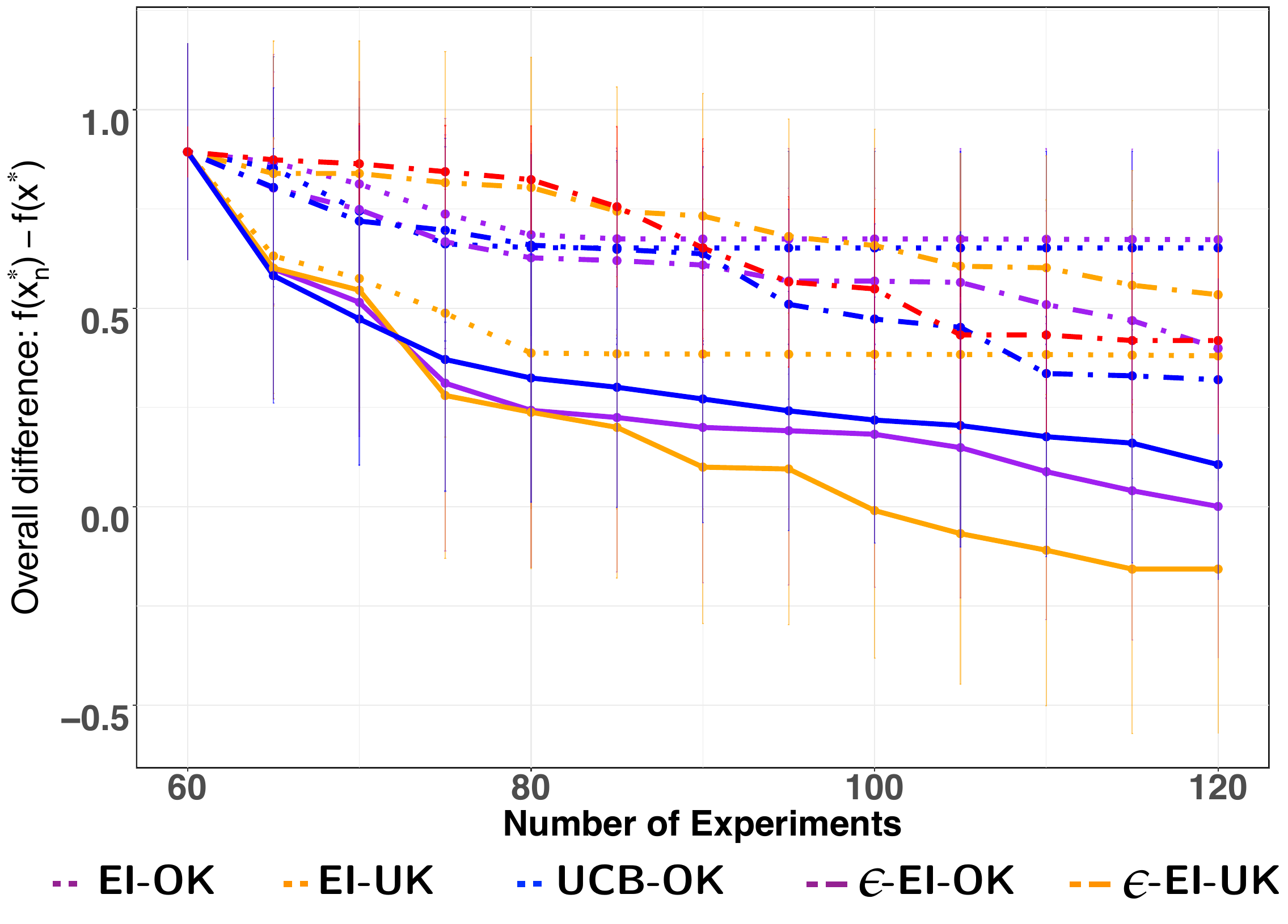}\label{fig:levy}}\hspace{0.34in}
\subfigure[Ackley (10-d)]
{\includegraphics[width = .4\textwidth]{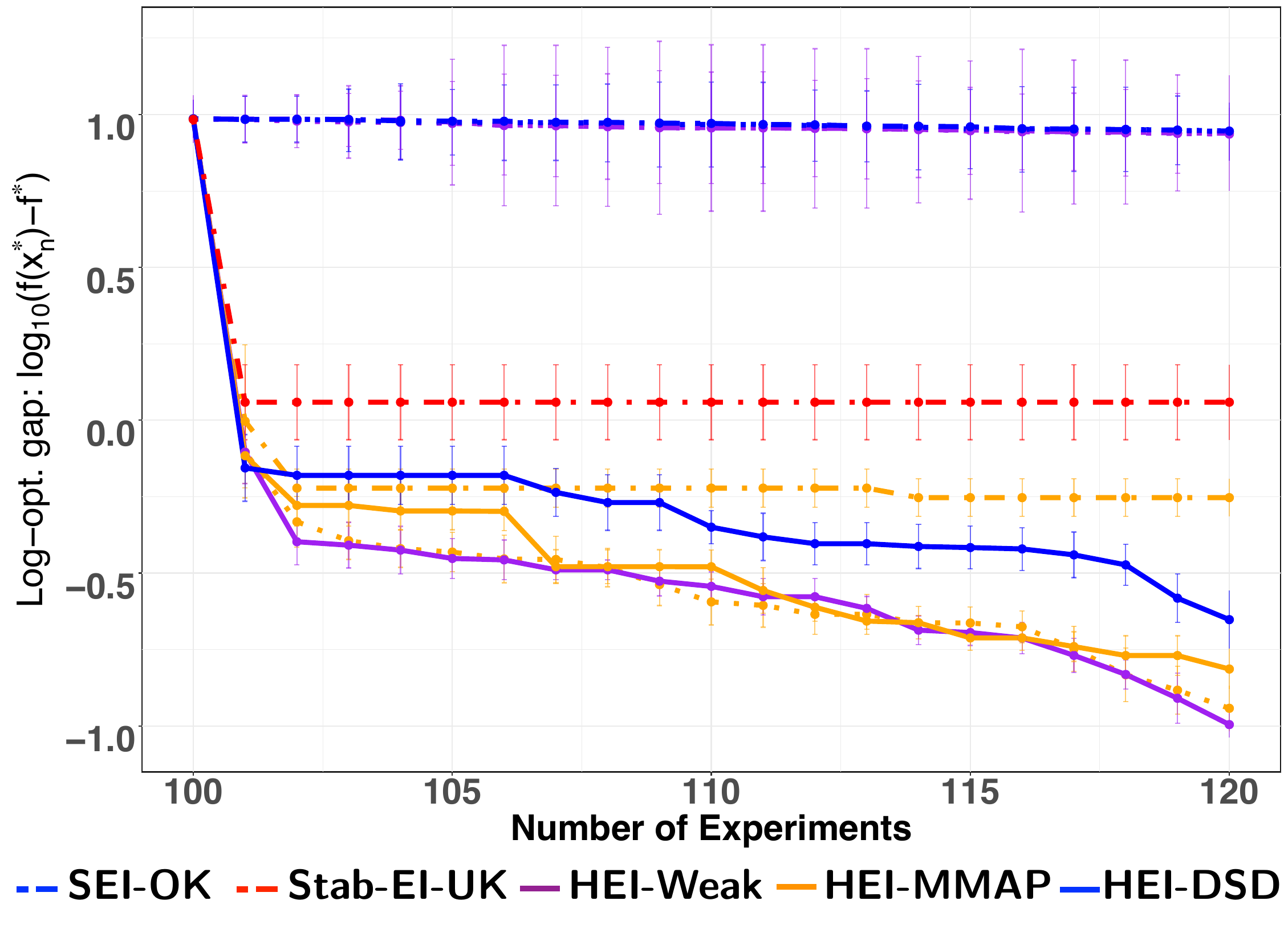}\label{fig:ackley}}
\caption{Numerical results for synthetic functions. (a) and (c)-(f) show the average optimality gap over 10 replications (dotted lines: EI-OK, EI-UK, and UCB-OK; dashed lines: $\epsilon$-EI-OK, $\epsilon$-EI-UK, SEI-OK, Stab-EI-UK; solid lines: HEI-Weak, HEI-MMAP, HEI-DSD). (b) presents a visualization of sampled points for the Branin function (grey squares: initial points, black stars: global optima, red triangles: UCB-OK points, blue circles: HEI-DSD points). }
\label{fig:EI_bra}
\end{center}
\end{figure}

Figures \ref{fig:best}-\ref{fig:hump6:best} and \ref{fig:ackley} show the log-optimality gap $\log_{10} (f(\bx_n^*)-f(\bx^*))$ against the number of samples $n$ for the first three functions, and Figure~\ref{fig:levy} shows the optimality gap $f(\bx_n^*) - f(\bx^*)$ for the Levy function. We see that the three HEI methods outperform the existing Bayesian optimization methods: the optimality gap for the latter methods stagnates for larger sample sizes, whereas the former enjoys steady improvements as $n$ increases. This shows that the proposed method indeed corrects the over-greediness of EI, and provides a more effective correction of this via hierarchical modeling, compared to the $\epsilon$-greedy and Stab-EI-UK methods. Furthermore, of the HEI methods, HEI-MMAP and HEI-DSD appear to greatly outperform HEI-Weak. This is in line with the earlier observation that weakly informative priors may yield poor optimization for HEI; the MMAP and DSD specifications give better performance by mimicking a fully Bayesian optimization procedure. {This is further supported by Figure~\ref{fig:HEI-post-variance}, which shows the posterior estimate $\tilde{\sigma}_n$ of the scale parameter as a function of sample size $n$ for the HEI methods in the Branin experiment. We see that both the HEI-DSD and HEI-MMAP provide larger estimates than HEI-DSD, which shows that the former methods are indeed integrating further uncertainty for exploration.} The steady improvement of HEI-DSD also supports the data-size-dependent prior condition needed for global convergence in Theorems \ref{thm:con} and \ref{thm:stab}. 

Figure~\ref{fig:visual} shows the sampled points from HEI-DSD and UCB-OK for one run of the Branin function. The points for HEI-Weak and HEI-MMAP are quite similar to HEI-DSD, and the points for EI-OK, EI-UK and SEI are quite similar to UCB-OK, so we only plot one of each for easy visualization. We see that HEI indeed encourages more exploration in optimization: it successfully finds \textit{all} three global optima for $f$, whereas existing methods cluster points near only one optimum. The need to identify multiple global optima often arises in multiobjective optimization. For example, a company may wish to offer multiple product lines to suit different customer preferences \citep{mak2019analysis}. For such problems, HEI can provide more practical solutions over existing methods.

\vspace{.05in}

Lastly, we compare the performance of HEI with the SEI method \citep{benassi2011robust}. From Figure~\ref{fig:levy}, we see that the SEI performs quite well for the Levy function: it is slightly worse than HEI methods, but better than the other methods. However, from Figure~\ref{fig:EI_bra}, the SEI achieves only comparable performance with EI-OK for the Branin function (which is in line with the results reported in \cite{benassi2011robust}), and is one of the worst-performing methods. This shows that the performance of SEI can vary greatly for different problems. 

{Finally, we note that for higher-dimensional problems, the HEI (and indeed, all GP-based Bayesian optimization methods) will likely require more sophisticated GP models to work well. In particular, such GP models should ideally be able to learn low-dimensional embeddings of the objective function over the high-dimensional domain; see, e.g., recent work in \cite{Seshadri2019, zhang2022gaussian}. Integrating such low-dimensional structure within the HEI will be the topic of future work.}

\section{Semiconductor Manufacturing Optimization}\label{sec:app}

We now investigate the performance of HEI in a process optimization problem in semiconductor wafer manufacturing. In semiconductor manufacturing  \citep{jin2012sequential}, thin silicon wafers undergo a series of refinement stages. Of these, thermal processing is one of the most important stage,
since it facilitates necessary chemical reactions and allows for surface oxidation \citep{singh2000rapid}. Figure~\ref{fig:pro} visualizes a typical thermal processing procedure: a laser beam is moved radially in and out across the wafer, while the wafer itself is rotated at a constant speed. There are two objectives here. First, the wafer should be heated to a target temperature to facilitate the desired chemical reactions. Second, temperature fluctuations over the wafer surface should be made as small as possible, to reduce unwanted strains and improve wafer fabrication~\citep{brunner2013characterization}. The goal is to find an ``optimal'' setting of the manufacturing process which achieves these two objectives.

We consider five control parameters: wafer thickness, rotation speed, laser period, laser radius, and laser power (a full specification is given in Table~\ref{table:range}). The heating is performed over 60 seconds, and a target temperature of $\mathcal{T}^*=600$ F is desired over this timeframe. We use the following objective function:
\begin{align}
    f(\bx): = \sum_{t=1}^{60} \max_{\mathbf{s}\in\mathcal{S}}|\mathcal{T}_t(\mathbf{s};\bx)-\mathcal{T}^*|.
    \label{eqn:fnlaser}
\end{align}
Here, $\mathbf{s}$ denotes a spatial location on the wafer domain $\mathcal{S}$, $t = 1, \cdots, 60$ denotes the heating time (in seconds), and $\mathcal{T}_t(\mathbf{s};\bx)$ denotes the wafer temperature at location $\mathbf{s}$ and time $t$, using control setting $\bx \in \mathbb{R}^5$. Note that $f(\bx)$ captures both objectives of the study: wafer temperatures $\mathcal{T}_t$ close to $\mathcal{T}^*$ results in smaller values of $f(\bx)$, and the same is true when $\mathcal{T}_t(\mathbf{s};\bx)$ is stable over $\mathbf{s} \in \mathcal{S}$.

\begin{figure}[!tbh]
    \centering
\includegraphics[width = 0.45\textwidth]{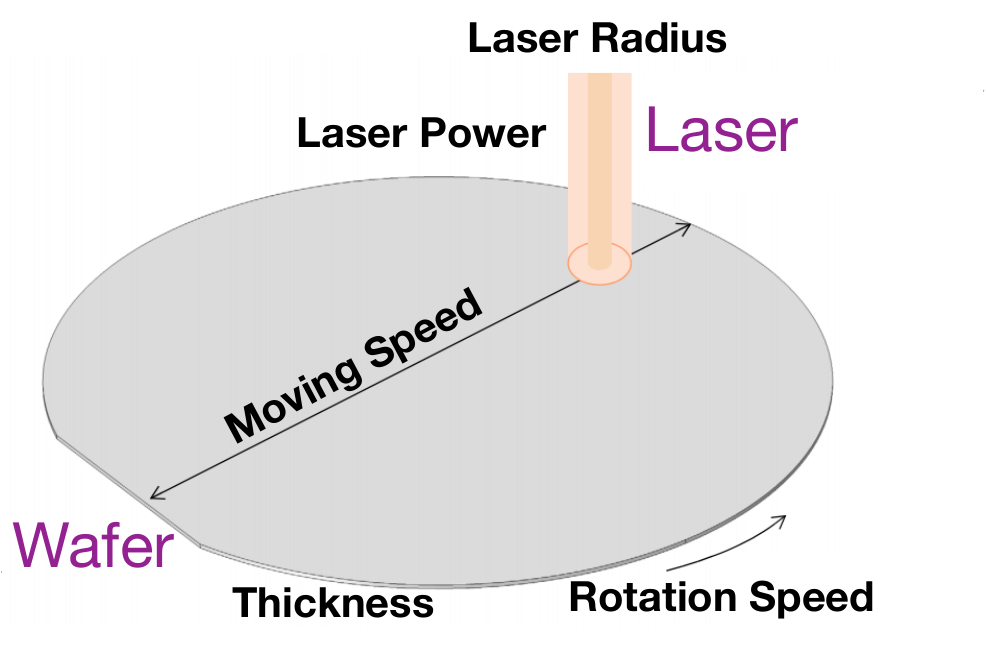}\label{fig:pro}
\includegraphics[width = 0.45\textwidth]{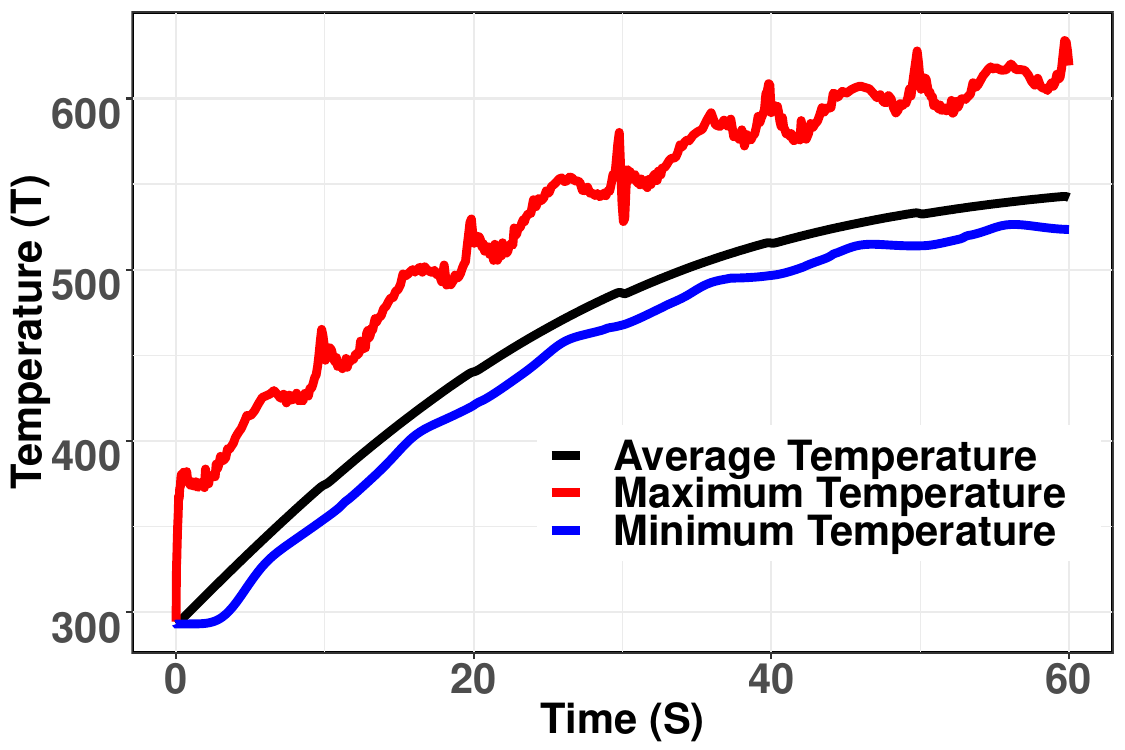}\label{fig:example}
    \caption{A visualization of the laser heating process of a silicon wafer (left), and the simulated temperature profile from COMSOL Multiphysics (right).}
    \label{fig:my_label}
\end{figure}

Clearly, each evaluation of $f(\bx)$ is expensive, since it requires a full run of wafer heating process. We will simulate each run using COMSOL Multiphysics \citep{comsol}, a reliable finite-element analysis software for solving complex systems of partial differential equations (PDEs). COMSOL models the incident heat flux from the moving laser as a spatially distributed heat source on the surface, then computes the transient thermal response by solving the coupled heat transfer and surface-to-ambient radiation PDEs. Figure~\ref{fig:example} visualizes the simulation output from COMSOL: the average, maximum, and minimum temperature over the wafer domain at every time step. Experiments are performed on a desktop computer with quad-core Intel I7-8700K processors, and take around $5$ minutes per run.

\begin{table}[!t]
\begin{center}
\caption{\label{table:range}Design ranges of the five control parameters, where rpm (revolutions per minute) measures the rotation speed of the wafer.}
\vspace{0.05in}
\begin{tabular}{c| c| c| c| c}
\hline
\hline
Thickness & Rotation Speed& Laser Period & Laser Radius & Power \\
\hline
$[160,300]\mu$s & $[2,50]$rpm & $[5,15]$s & $[2,10]$mm & $[10,20]$W\\
\hline
\hline
\end{tabular}
\end{center}
\vspace{-0.15in}
\end{table}

Figure~\ref{fig:laser_bo} shows the best objective values $f(\bx_n^*)$ for HEI-MMAP and HEI-DSD (the best performing HEI methods from simulations), and for the UCB-OK, SEI, and $\epsilon$-greedy EI methods. We see that UCB-OK and SEI perform noticeably poorly, whereas the proposed HEI-MMAP and HEI-DSD methods provide the best optimization performance, with the $\epsilon$-greedy-EI method slightly worse. This again shows that the proposed HEI can provide a principled correction to the over-greediness of EI via hierarchical modeling, which translates to more effective optimization performance over the compared existing methods.

\begin{figure}[!t]
    \centering
    \subfigure[Best objective values]
{\includegraphics[width = .4\textwidth]{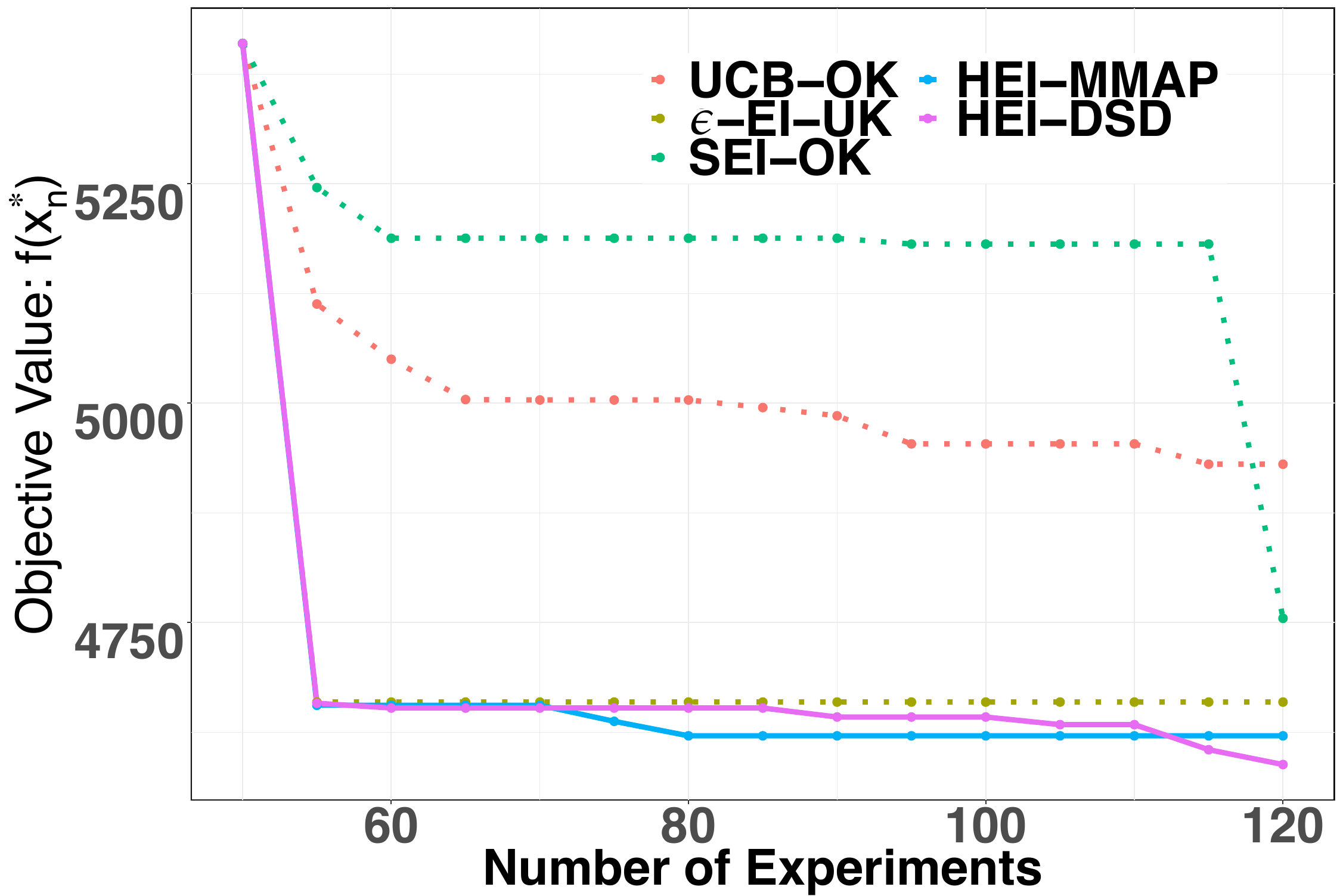}\label{fig:laser_bo}}\hspace{0.2in}
\subfigure[SEI]
{\includegraphics[width = .4\textwidth]{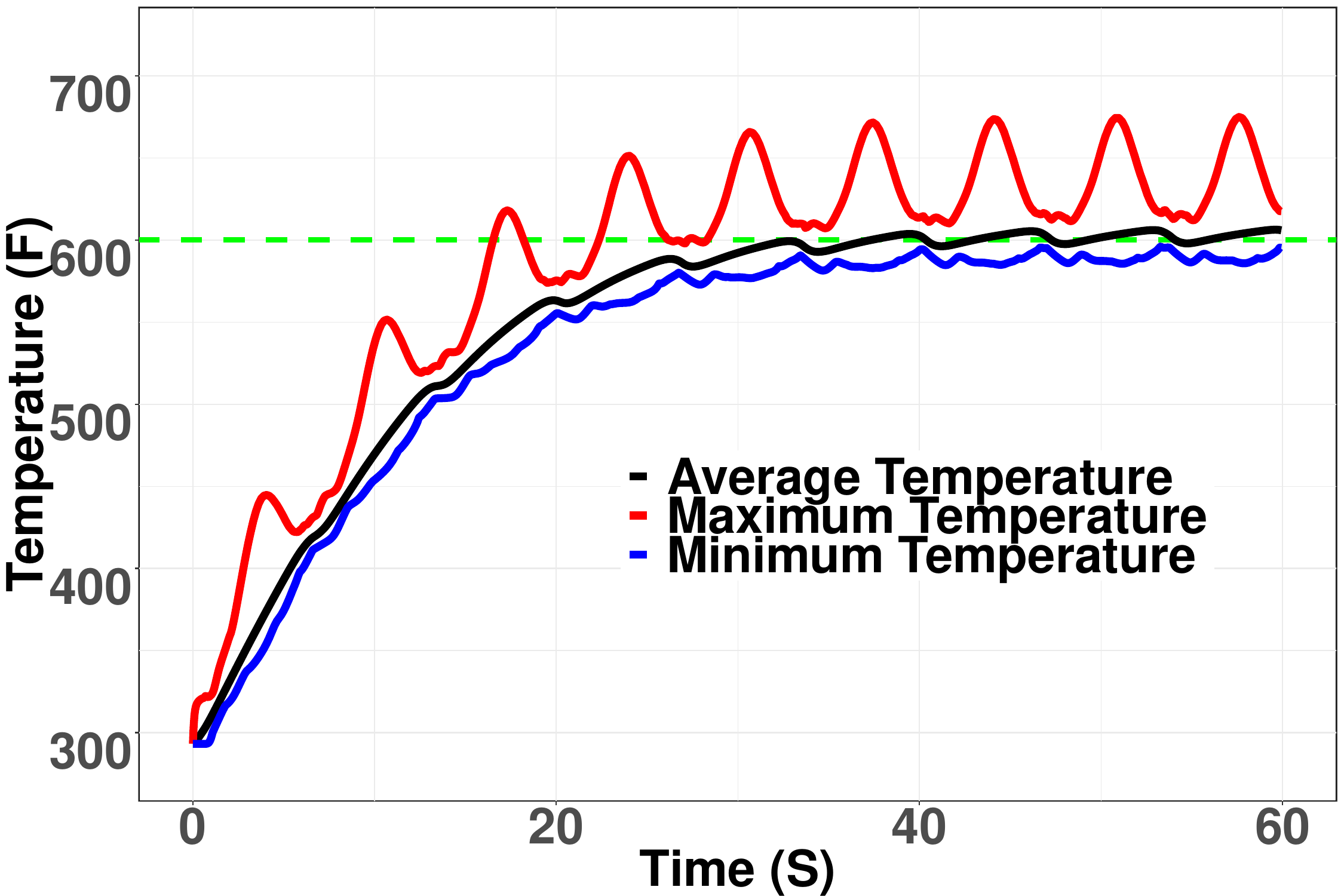}\label{fig:sei}}\\
\subfigure[HEI-MMAP]
{\includegraphics[width = .4\textwidth]{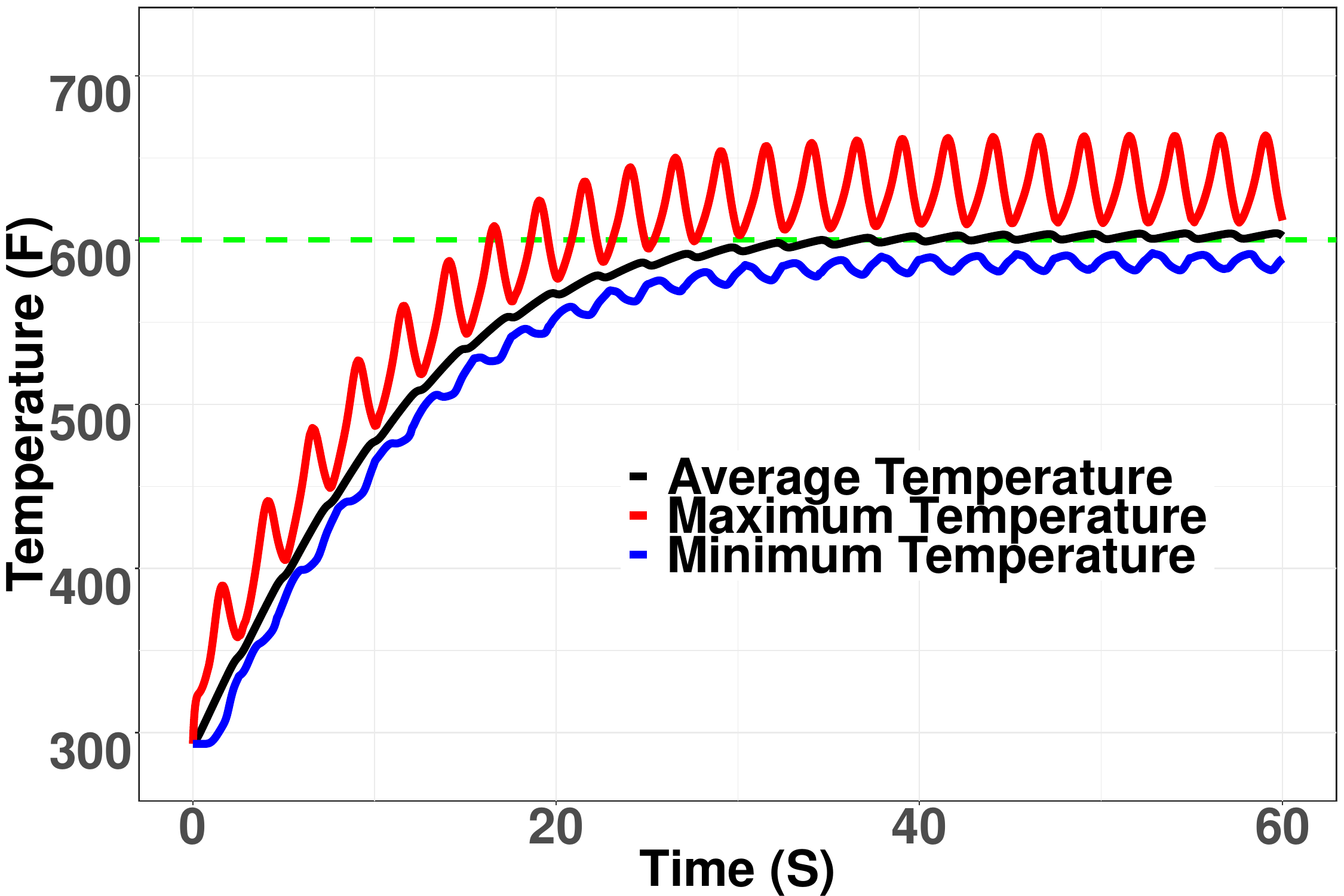}\label{fig:mmap}}\hspace{0.2in}
\subfigure[$\epsilon$-EI-UK]
{\includegraphics[width = .4\textwidth]{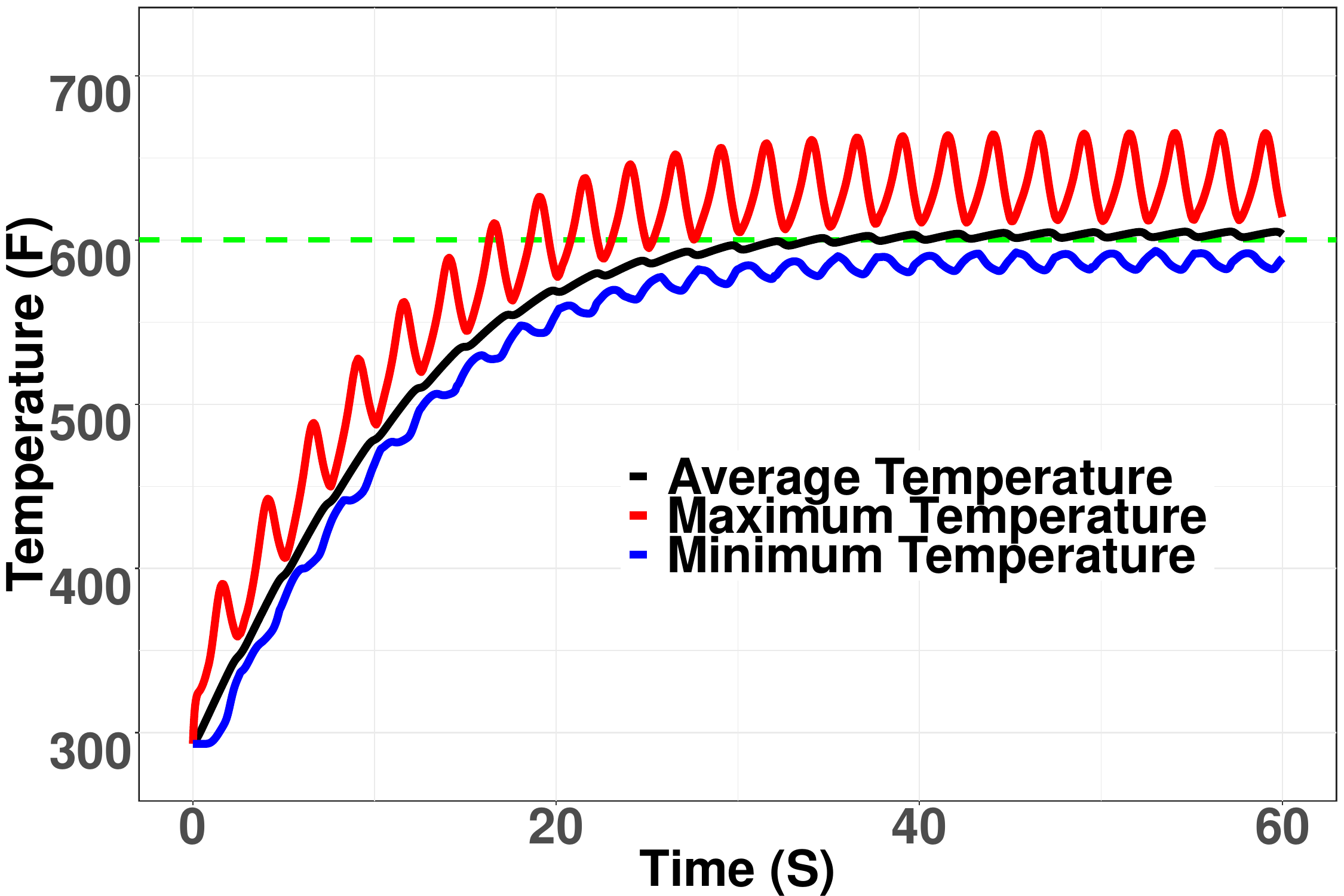}\label{fig:greedy}}\\
\subfigure[HEI-DSD]
{\includegraphics[width = .4\textwidth]{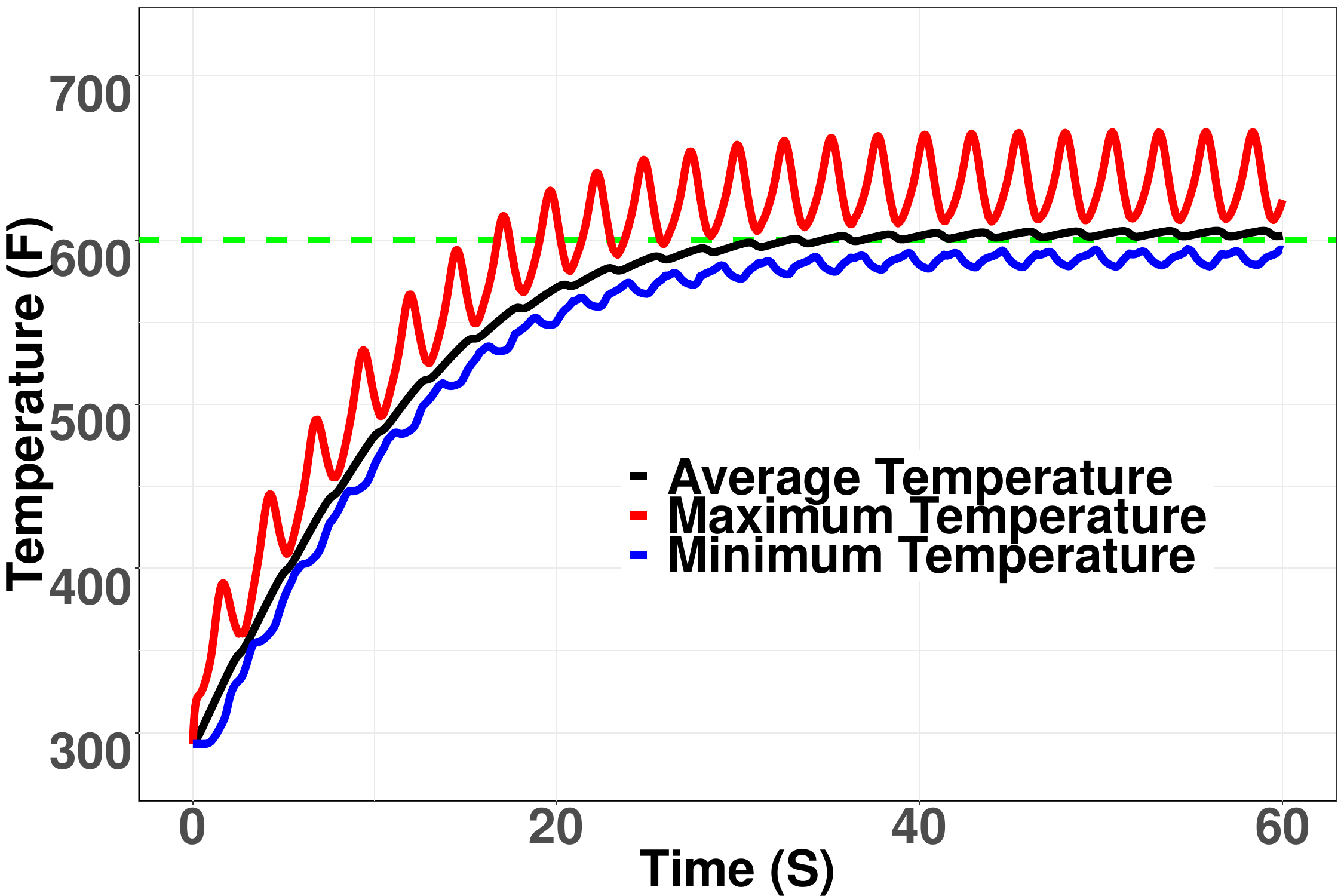}\label{fig:dsd}}\hspace{0.2in}
\subfigure[UCB-OK]
{\includegraphics[width = .4\textwidth]{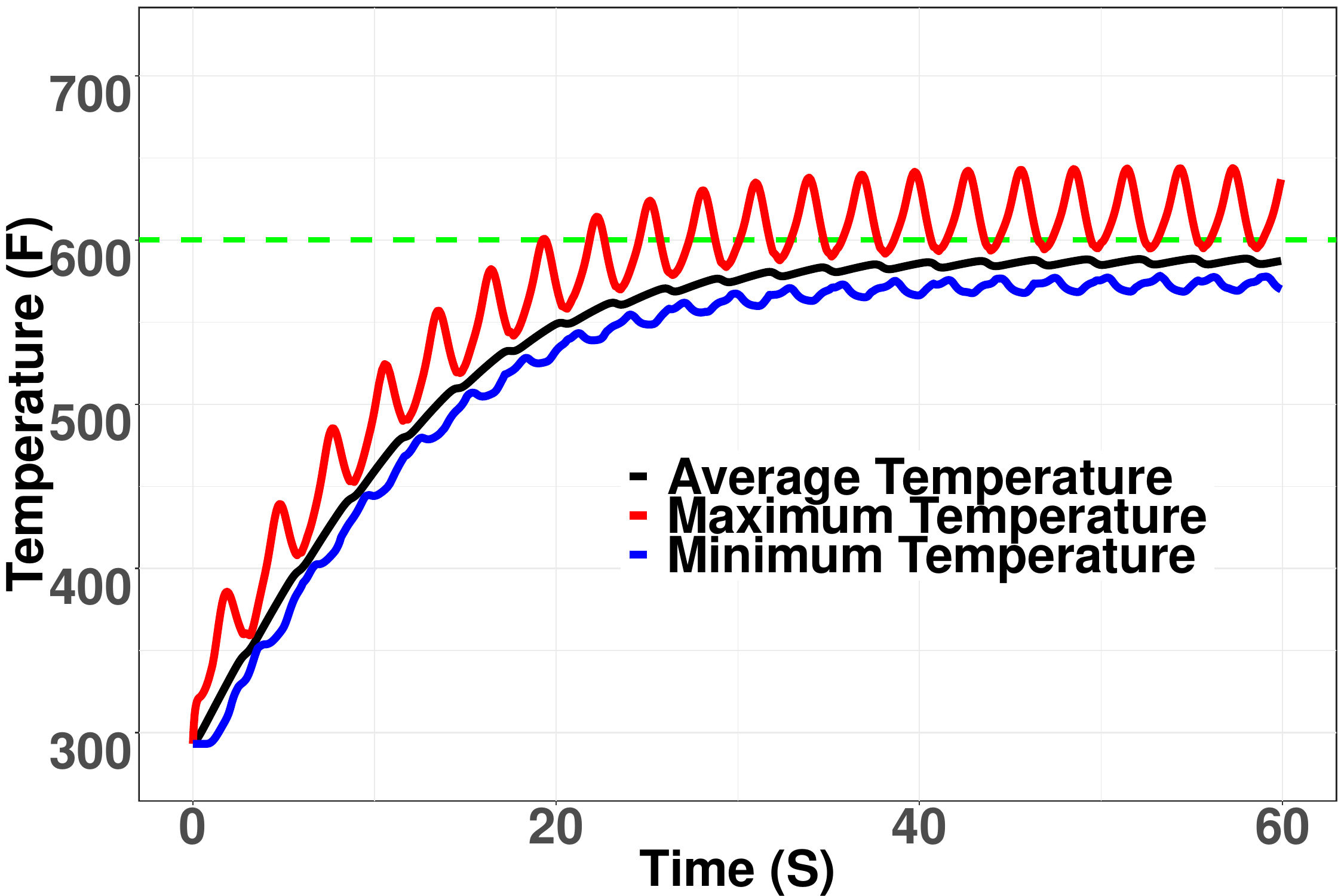}\label{fig:ucb}}
    \caption{(a) shows the best objective value $f(\bx^*_n)$ for the five compared methods. (b)-(f) show the average, maximum, and minimum temperature of the wafer over time, for each of the tested BO methods. The dotted green line marks the target temperature of $T^*=$ 600 F. }
    \label{fig:tem}
\end{figure}
The remaining plots in Figure~\ref{fig:tem} show the average, maximum, and minimum temperature over the wafer surface, as a function of time. For HEI-DSD and HEI-MMAP, the average temperature quickly hits 600 F, with a slight temperature oscillation over the wafer. For SEI, the average temperature reaches the target temperature slowly, but the temperature fluctuation is much higher than for HEI-DSD and HEI-MMAP. For UCB-OK, the average temperature does not even reach the target temperature. The two proposed HEI methods (and $\epsilon$-EI-UK, although its performance is slightly worse) return noticeably improved settings compared to the two earlier methods, thereby providing engineers with an effective and robust wafer heating process for semiconductor manufacturing.

\section{Conclusion}\label{sec:conc}
In this paper, we presented a hierarchical expected improvement (HEI) framework for Bayesian optimization of a black-box objective $f$. HEI aims to correct a key limitation of the expected improvement (EI) method: its over-exploitation of the fitted GP model, which results in a lack of convergence to a global solution even for smooth objective functions. HEI addresses this via a hierarchical GP model, which integrates parameter uncertainty of the fitted model within a closed-form acquisition function. This provides a principled way for correcting over-exploitation by encouraging exploration of the optimization space. We then introduce several hyperparameter specification methods, which allow HEI to efficiently approximate a fully Bayesian optimization procedure. Under certain prior specifications, we prove the global convergence of HEI over a broad function class for $f$, and derive near-minimax convergence rates. In numerical experiments, HEI provides improved optimization performance over existing Bayesian optimization methods, for both simulations and a process optimization problem in semiconductor manufacturing.

{Given these promising results, there are several intriguing avenues for future work. One direction is to explore potential extensions of the HEI for the high-dimensional setting of $d > n$, where the dimension of the problem may exceed the number of function evaluations. This presents interesting challenges for both theory and methodology, and may require more sophisticated GP models which can learn low-dimensional embeddings in high dimensions (see, e.g., \citealp{zhang2022gaussian}). Another interesting direction is to study the dependence of the optimization rates on the constant $\gamma$ in Condition \ref{ass:stb}.}

\bibliographystyle{chicago}
\vspace{-0.1in}
{\setlength{\bibsep}{0.0pt}  \bibliography{ref}}

\begin{appendix}

\section{Proofs}
\subsection{Proof of Proposition 2}\label{pro:thm:nei}

\begin{proof}
By Lemma 1, the posterior distribution follows a non-standardized t-distribution:
$$\Big[ f(\bx)\big| \cD_n \Big] \sim T \Big(2a+n-q, \hat{f}_n(\bx), \tilde{\sigma}_ns_n(\bx)\Big).$$
Let $\nu_n=2a+n-q$. The density function of $\Big[ f(\bx)\big| \cD_n \Big]$ then takes the following form:
$$g(f;\nu_n,\hat{f}_n,\tilde{\sigma}_n,s_n) = \frac{\Gamma((\nu_n+1)/2)}{\tilde{\sigma}_ns_n\sqrt{\nu_n\pi}\cdot\Gamma(\nu_n/2)}\left( 1 + \frac{(f-\hat{f}_n)^2}{\nu_n \tilde{\sigma}_n^2s_n^2}\right)^{-(\nu_n+1)/2}.$$
Using this density function, the HEI criterion can then be simplified as:
\begin{align}\label{eqn:dri}
\textstyle\HEI_n(\bx) & = \EE_{f|\cD_n}(y_n^*- f(\bx))_+  \notag\\
& = \left(y_n^*-\hat{f}_n\right)\Phi_{\nu_n}\left( \frac{y_n^*-\hat{f}_n}{\tilde{\sigma}_ns_n}\right) + \int_{-\infty}^{y_n^*} \left(\hat{f}_n-f\right) g(f) df.
\end{align}
The second term in \eqref{eqn:dri} can be further simplified as:
\begin{align}
&\int_{-\infty}^{y_n^*} \left(\hat{f}_n-f\right) g(f) df\\
=& - \frac{\tilde{\sigma}_ns_n}{2}\int_{-\infty}^{\left(\frac{y_n^*-\hat{f}_n}{\tilde{\sigma}_ns_n}\right)^2}\frac{\Gamma(({\nu_n}+1)/2)}{\sqrt{{\nu_n}\pi}\cdot\Gamma({\nu_n}/2)}\left( 1 + \frac{t}{{\nu_n}}\right)^{-\frac{\nu_n+1}{2}}dt\\ 
=&  \int_{-\infty}^{y_n^*} \left(\hat{f}_n-f\right)\frac{\Gamma((\nu_n+1)/2)}{\tilde{\sigma}_ns_n\sqrt{\nu_n\pi}\cdot\Gamma(\nu_n/2)}\left( 1 + \frac{(f-\hat{f}_n)^2}{\nu_n \tilde{\sigma}_n^2s_n^2}\right)^{-(\nu_n+1)/2}df\\
=&  \frac{\tilde{\sigma}_ns_n{\nu_n}}{{\nu_n}-1} \frac{\Gamma(({\nu_n}+1)/2)}{\sqrt{{\nu_n}\pi}\cdot\Gamma({\nu_n}/2)}\left(1+\frac{(f-\hat{f}_n)^2}{{\nu_n} (\tilde{\sigma}_ns_n)^2}\right)^{-\frac{\nu_n-1}{2}}\Bigg|_{-\inf}^{y_n^*} \\
=&  \frac{\tilde{\sigma}_ns_n{\nu_n}}{{\nu_n}-1} \frac{\Gamma(({\nu_n}+1)/2)}{\sqrt{{\nu_n}\pi}\cdot\Gamma({\nu_n}/2)}\left(1+\frac{(y_n^*-\hat{f}_n)^2}{{\nu_n} (\tilde{\sigma}_ns_n)^2}\right)^{-\frac{\nu_n-1}{2}}\hspace{1in}\textrm{(Since $v_n > 1$)}\\
=& \frac{\sqrt{{\nu_n}}\tilde{\sigma}_ns_n\Gamma(({\nu_n}-1)/2)}{({\nu_n}-2)\sqrt{\pi}\cdot\Gamma(({\nu_n}-2)/2)}\left(1+ \frac{1}{{\nu_n}}\left(\frac{y_n^*-\hat{f}_n}{\tilde{\sigma}_ns_n}\right)^2\right)^{-\frac{\nu_n-1}{2}}\\
=&\sqrt{\frac{{\nu_n}}{{\nu_n}-2}}\tilde{\sigma}_ns_n\phi_{{\nu_n}-2}\left(\frac{y_n^*-\hat{f}_n}{\sqrt{{\nu_n}/({\nu_n}-2)}\tilde{\sigma}_ns_n}\right).
\end{align}
Therefore, we prove the claim.
\end{proof}

\subsection{Proof of Proposition 3}\label{pro:prop:sei}
If necessary, we denote the correlation function by $s_n(\bx;\boldsymbol\theta)$ to highlight the dependence of $\boldsymbol\theta$. If not specified, we use simplified $s_n(\bx)$ to make general arguments. 

The proof uses several Lemmas in \cite{bull2011convergence}. The following lemma provides a lower bound for $s_n(\bx;\boldsymbol\theta)$. 

\begin{lemma}[Lemma 7 in \cite{bull2011convergence}]\label{lem:upper}
Set $\zeta = \alpha$ if $\nu\leq 1$ otherwise 0. Given $\boldsymbol\theta \in \RR^d_+,$ there is a constant $C' > 0$ depending only on $\Omega, K$ and $\boldsymbol\theta$ which satisfies the following:\\
For any $k \in \NN$, and sequences $\bx_n \in \Omega$, $\tilde{\boldsymbol\theta}_n \geq \boldsymbol\theta$, the inequality
$$s_n(\bx_{n+1};\tilde{\boldsymbol\theta}_n) \geq C'k^{-(\nu\land 1)/d}(\log k)^\zeta$$
holds for at most $k$ distinct $n$.
\end{lemma}

\begin{lemma}[Lemma 9 in \cite{bull2011convergence}]\label{lemma:sup}
Given $\boldsymbol\theta^L,\boldsymbol\theta^U\in\RR^d_+$, pick sequences $\bx_n \in \Omega$, the corresponding posterior of scale parameter $\boldsymbol\theta^L\leq \tilde{\boldsymbol\theta}_n\leq \boldsymbol\theta^U.$ Then for an open $S \subset \Omega,$
\begin{align}
    \sup_{x\in S} s_n(x;\tilde{\boldsymbol\theta}_n) = \boldsymbol\Omega(n^{-\nu/d}),
\end{align}
uniformly in the sequences $x_n, \tilde{\boldsymbol\theta}_n$. Here $\boldsymbol\Omega$ denotes the asymptotic lower bound notation. 
\end{lemma}

\begin{proof}
This is a constructive proof. The key idea is that given the initial points $\bx_1,\cdots,\bx_k$ independent of the objective function, we construct two functions $h(\bx)$ and $\tilde{h}(\bx)$ such that SEI cannot distinguish these two functions and misses the global optimal of $\tilde{h}$.

First, given a random initial strategy $F$, we can use a union of small open subsets, denoted by $\mathcal{X}_0$, such that $k$ initial points generated by the strategy satisfies that $P_F(\bx_1,...,\bx_k\in\mathcal{X}_0)\leq \epsilon$. Then we partition the domain $\Omega$ as shown in Figure~\ref{fig:partition}. Here $\mathcal{X}_0$ and $\mathcal{Y}_0$ are two disjoint non-empty interior domains. 
\begin{figure}[!t]
    \centering
    \includegraphics[width = 0.5\textwidth]{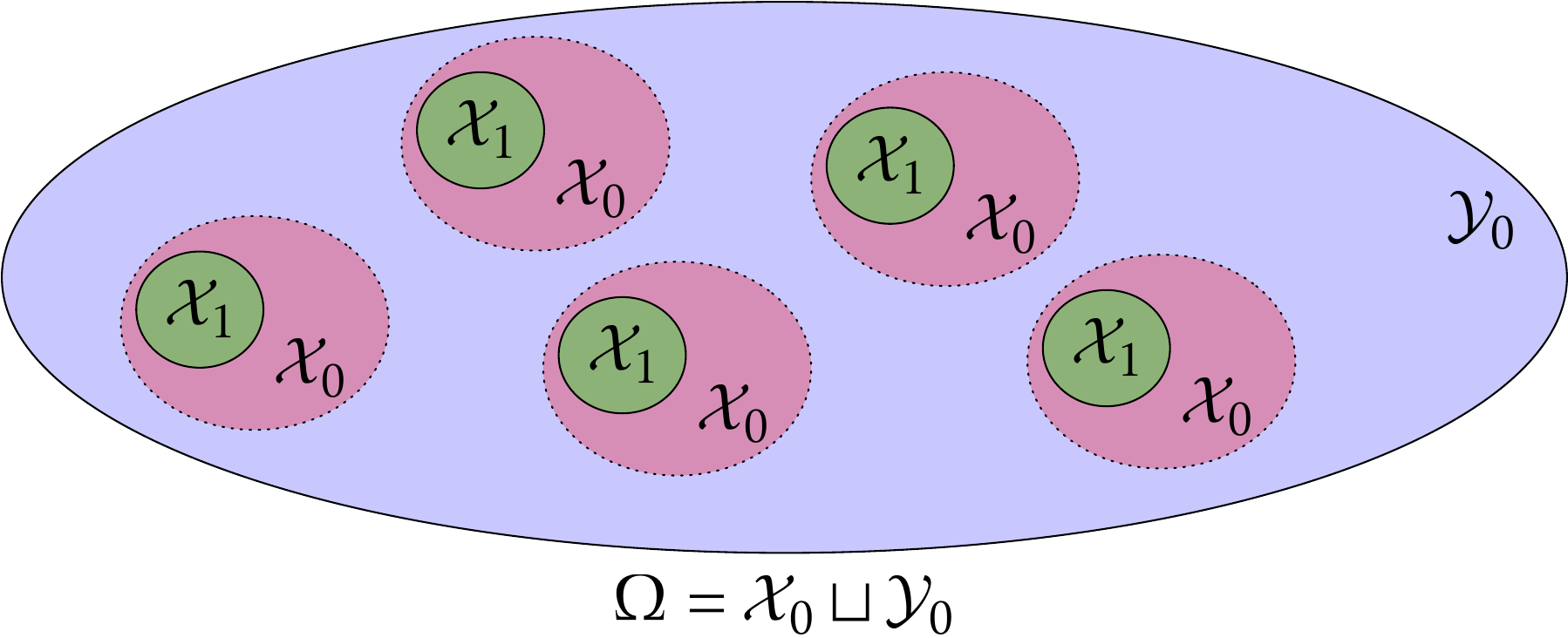}
    \caption{An illustrative example of how domain $\Omega$ is partitioned to two disjoint parts: an open set $\mathcal{X}_0$ and a closed set $\mathcal{Y}_0$. $\mathcal{X}_1$ is a closed subset of $\mathcal{X}_0$ (each green circle is a subset of $\mathcal{X}_1$ and each pink circle is a subset of $\mathcal{X}_0$). Under the distribution $F$ over $\Omega$, the probability of choosing $\mathcal{X}_0$ is less than $\epsilon$.}
    \label{fig:partition}
\end{figure}
Then we first construct a function $h$ as follows: 
\begin{align}
h(\bx) := \left\{
\begin{array}{cc}
    0 & \bx\in \mathcal{Y}_0 \\
    1 & \bx\in\mathcal{X}_1\\
    h_1(\bx) & \bx\in \mathcal{X}_0\backslash\mathcal{X}_1
\end{array}
\right..
\end{align}
Here $h_1(\bx)\geq 0$ is a function that ensures $h(\bx)\in C^\infty(\Omega)$. Thus, it is easy to verify that $h(\bx) \in \mathcal{H}_{\boldsymbol\theta}(\Omega)$, {since $\Omega$ is compact}. We denote $\norm{h}_{\mathcal{H}_{\boldsymbol\theta}(\Omega)} = R^2$. With such a function $h(\bx)$, if there is one $i\in\{1,2,...,k\}$ such that $\bx_i\in\mathcal{Y}_0$, then $y_k^* = 0$. Therefore, the probability of $y_k^* = 0$ is at least $1-\epsilon$. Then we show that conditioning on $y_k^*=0$, denoted by event $A_0$, SEI cannot visit $\mathcal{X}_1$ infinitely often. We prove this by contradiction. By Lemma~\ref{lem:upper}, we know that as $n\rightarrow \infty$, $s_n(\bx_{n+1};\tilde{\boldsymbol\theta}_n)\rightarrow 0.$ Suppose $\bx_{n+1}\in\mathcal{X}_1$. If $s_n(\bx_{n+1};\tilde{\boldsymbol\theta}_n)=0$, then $\bx_{n+1}\in\{\bx_1,...,\bx_n\}$ and $\textrm{SEI}_n(\bx_{n+1})=0$. Therefore, we can find a $\bz_{n+1}\in \mathcal{Y}_0$ such that $s_n(\bz_{n+1})\neq 0$ and $\textrm{SEI}_n(\bx_{n+1})=0< \textrm{SEI}_n(\bz_{n+1})$. 
If $s_n(\bx_{n+1};\tilde{\boldsymbol\theta}_n)>0$, then we have $$T_n(\bx_{n+1}):=\frac{y_n^*-h(\bx_{n+1})}{s_n(\bx_{n+1};\tilde{\boldsymbol\theta}_n)} = -s_n(\bx_{n+1};\tilde{\boldsymbol\theta}_n)^{-1} \rightarrow -\infty.$$
SEI~\citep{benassi2011robust} has the following form:
\begin{align}
\label{eqn:sei_pro}
    \textrm{SEI}_n(\bx) =  \tilde{\sigma}_ns_n(\bx)\Big(\frac{I_n(\bx)}{\tilde{\sigma}_ns_n(\bx)}\Phi_n\big(\frac{I_n(\bx)}{\tilde{\sigma}_ns_n(\bx)}\big)+\frac{\nu_n+(I_n(\bx)/\tilde{\sigma}_ns_n(\bx))^2}{\nu_n-1}\phi_n\big(\frac{I_n(\bx)}{\tilde{\sigma}_ns_n(\bx)}\big)\Big).
\end{align}
Moreover, $T_n$ and $I_n$ are connected by Lemma \ref{lem:err}, i.e., $$I_n(\bx)/s_n(\bx)\leq T_n(\bx)+\norm{h}_{\mathcal{H}_{\tilde{\boldsymbol\theta}_n}(\Omega)}\leq T_n(\bx)+R\sqrt{\prod_{i=1}^d\theta_i^U/\theta_i^L}=:T_n(\bx)+S.$$ The last inequality holds because of Lemma~\ref{lem:bounded-norm}. Then we establish an upper bound for the part in parentheses in \eqref{eqn:sei_pro}:
\begin{align}
\eta \Phi_n(\eta)+\frac{\nu_n+\eta^2}{\nu_n-1}\phi_n(\eta) &= \Theta\left( -\nu_n^{-(\nu_n+1)/2}\frac{\eta^{-\nu_n-1}}{\nu_n+2} + 2  \nu_n^{-(\nu_n+1)/2}\frac{|\eta|^{-\nu_n+1}}{\nu_n} \right)\notag\\
&= \Theta(|\eta|^{-\nu_n+1}) \quad\textrm{ as $\eta \rightarrow -\infty$.}
\end{align}
The first equality holds because of the fact $(1+\eta^2/\nu_n)\geq \frac{x^2}{\nu_n}.$
On the other hand, by Lemma~\ref{lemma:sup}, there exists a $\bz_{n+1}\in \mathcal{Y}_0$ such that $h(\bz_{n+1}) = 0$ and $s_{n}(\bz_{n+1};\tilde{\boldsymbol\theta}_n)=\boldsymbol\Omega (n^{-\nu/d}).$ Then we have $$\textrm{SEI}_n(\bz_{n+1};\tilde{\boldsymbol\theta}_n)=\tilde{\sigma}_ns_n(\bz_{n+1};\tilde{\boldsymbol\theta}_n)\left(\frac{\nu_n}{\nu_n-1} \frac{\Gamma((\nu_n+1)/2)}{\sqrt{\nu_n\pi}\Gamma(\nu_n/2)}\right).$$
Here $\Gamma$ denotes the gamma function. Therefore, we obtain
\begin{align}
    \frac{\textrm{SEI}_n(\bx_{n+1};\tilde{\boldsymbol\theta}_n)}{\textrm{SEI}_n(\bz_{n+1};\tilde{\boldsymbol\theta}_n)}\leq\mathcal{O}\left(\frac{((T_n(\bx)+S)/\tilde{\sigma}_n)^{-\nu_n+1}}{n^{-\nu/d}}\right) = o(1).
\end{align}
The last equality holds because $T_n\rightarrow -\infty$ and $\nu_n$ is the same order as $n$ under fixed hyperparameters. {Then, given any positive value $A$, there exists an integer $n_A$ such that for any $n > n_A$, $T_n < -A$. Moreover, since $\nu_n$ is of the same order as $n$, then $\big|\frac{((T_n(\bx)+S)/\tilde{\sigma}_n)^{-\nu_n+1}}{n^{-\nu/d}}\big| < \frac{A^{-n}}{n^{-\nu/d}} \rightarrow 0$ as $n \rightarrow \infty$}. Therefore, by contradiction, we know that on $A_0$, there is a random variable $w$, for all $n > w$, $x_n \not\in \mathcal{X}_1$. Hence there is a constant $t \in \NN$, depending on $\epsilon$, such that the event $A_1 = A_0\cap \{w \leq  t\}$ has probability at least $1-2\epsilon$ under the SEI strategy. Then we can further select an open set $\mathcal{W}\subset \mathcal{X}_1$ such that the event $A_2 = A_1\cap \{\bx_n\not\in \mathcal{W}:n<t \}$ has probability at least $1-3\epsilon$. 

Finally, we can construct another smooth function  $g(\bx)$ like $h(\bx)$, which equals $0$ when $\bx\not\in\mathcal{W}$ and has minimum $-2$.  Then we construct $\tilde{h}(\bx):=h(\bx) + g(\bx)$, which has minimum $-1$. However, SEI cannot distinguish the difference between these two functions and $\inf_n \tilde{h}(\bx_n)\geq 0$. Thus, for $\delta=1$, we have
$$\PP_F\left( \inf_{n}\tilde{h}(\bx_n^*) - \min_{\bx\in\Omega}\tilde{h}(\bx)\geq \delta\right)\geq1-3\epsilon.$$
We obtain the desired result.
\end{proof}

\subsection{Proof of Proposition 4}\label{pro:prop:eb}
\begin{proof}
The marginal likelihood can be obtained by integrating out the parameters $\boldsymbol{\beta}$ and $\sigma^2$ in the hierarchical model :
\begin{align}
&\textstyle p(\by_n;a, b)\notag\\
=& \int \frac{\exp\{-(\by_n-\mathbf{P}_n \boldsymbol{\beta})^\top\mathbf{K}_n^{-1}(\by_n-\mathbf{P}_n \boldsymbol{\beta})/(2\sigma^2)\}}{\sqrt{2\pi \det(\sigma^2 \mathbf{K}_n)}} \frac{(b/\sigma^2)^{a} }{\sigma^2\Gamma(a)} \exp\left(-\frac{b}{\sigma^2}\right)d \boldsymbol{\beta} d \sigma^2\notag\\
=& \int \sqrt{\frac{\det\left(\sigma^2 \mathbf{G}_n^{-1}\right)}{\det\left(\sigma^2 \mathbf{K}_n\right)}} \exp\left\{-\frac{\by_n^\top \mathbf{K}_n^{-1} \by_n - \hat{\boldsymbol{\beta}}_n^\top \mathbf{G}_n \hat{\boldsymbol{\beta}}_n}{2\sigma^2}\right\}\frac{(b/\sigma^2)^{a} }{\sigma^2\Gamma(a)} \exp\left(-\frac{b}{\sigma^2}\right) d \sigma^2\notag\\
=&  \sqrt{\frac{\det(\mathbf{G}_n^{-1})}{\det( \mathbf{K}_n)}} \frac{b^a}{\Gamma(a)}\int (\sigma^2)^{-(a+\frac{n-q}{2})-1}\exp\left\{-\frac{\by_n^\top \mathbf{K}_n^{-1} \by_n - \hat{\boldsymbol{\beta}}_n^\top \mathbf{G}_n \hat{\boldsymbol{\beta}}_n+2b}{2\sigma^2}\right\}d \sigma^2 \notag\\
=&    \sqrt{\frac{\det(\mathbf{G}_n^{-1})}{\det( \mathbf{K}_n)}} \frac{b^a}{\Gamma(a)} \frac{\Gamma(a+(n-q)/2)}{(b+(\by_n^\top \mathbf{K}_n^{-1} \by_n - \hat{\boldsymbol{\beta}}_n^\top \mathbf{G}_n \hat{\boldsymbol{\beta}}_n)/2)^{a+\frac{n-q}{2}}}.
\end{align}

Consider next the optimization of the marginal likelihood (17). Since the first term $\sqrt{\det(\mathbf{G}_n^{-1})/ \det(\mathbf{K}_n)}$ does not involve $a$ and $b$, we consider only the remaining terms in (17), and denote it as $p(\by_n;a, b)$. The partial derivative of $p(\by_n;a, b)$ in $b$ is:
\begin{align}
\textstyle \nabla_b \; p(\by_n;a, b) = \frac{\Gamma(a+(n-q)/2)}{\Gamma(a)}\frac{b^{a-1}(a(\by_n^\top \mathbf{K}_n^{-1} \by_n -\hat{\boldsymbol{\beta}}_n^\top \mathbf{G}_n \hat{\boldsymbol{\beta}}_n)/2- b(n-q)/2)}{(b+(\by_n^\top \mathbf{K}_n^{-1} \by_n - \hat{\boldsymbol{\beta}}_n^\top \mathbf{G}_n \hat{\boldsymbol{\beta}}_n)/2)^{a+(n-q)/2+1}}.
\end{align}
Setting this to zero and solving for $b$, we get the profile maximizer:
\begin{align}\label{eqn:opti}
\textstyle b^* (a)= a\left(\by_n^\top \mathbf{K}_n^{-1} \by_n -\hat{\boldsymbol{\beta}}_n^\top \mathbf{G}_n \hat{\boldsymbol{\beta}}_n\right)/(n-q).
\end{align} 

Now, let $w = \by_n^\top \mathbf{K}_n^{-1} \by_n - \hat{\boldsymbol{\beta}}_n^\top \mathbf{G}_n \hat{\boldsymbol{\beta}}_n$, in which case $b^*(a) = a \cdot w/(n-q)$. With this, the (rescaled) marginal likelihood can be written as a function of only $a$:
$$p( \by_n;a,b^*(a)) = \frac{(aw)^a\Gamma(a+(n-q)/2)}{\Gamma(a)(n-q)^a (a w/(n-q)+w/2)^{a+\frac{n-q}{2}}}.$$
Taking the gradient of $p( \by_n;a,b^*(a))$ in $a$, we get:
\begin{align}
\textstyle\nabla_a \; p( \by_n;a,b^*(a)) &= \frac{(aw)^a\Gamma(a+\frac{n-q}{2}) }{\Gamma(a) (a w/(n-q)+w/2)^{a+(n-q)/2}(n-q)^a}\notag \\
&\cdot \left(\Psi\left(a+\frac{n-q}{2}\right)-\Psi(a)-\log\left(1+\frac{n-q}{2a}\right)\right),
\end{align}
where $\Psi(x) = \frac{\Gamma'(x)}{\Gamma(x)}$ satisfies 
$\textstyle\Psi(a+1) = \frac{1}{a}+\Psi(a).$ 
Therefore, for even values of $n-q$, we have $$\Psi\left(a+\frac{n-q}{2}\right)-\Psi(a)= \sum_{i = 0}^{(n-q)/2-1}\frac{1}{a+i}\geq \log\left(1+\frac{n-q}{2a}\right),$$
while for odd values of $n-q$, we have  
\begin{align}
\Psi\left(a+\frac{n-q}{2}\right)-\Psi(a)&\geq \sum_{i = 0}^{(n-q-1)/2-1}\frac{1}{a+i}+\frac{1}{2a+n-q-1}\notag\\
&\geq \log\left(1+\frac{n-q-1}{2a}\right)+\frac{1}{2a+n-q-1}\notag\\
&\geq \log\left(1+\frac{n-q}{2a}\right).
\end{align}
Hence, $p(\by_n;a, b^*(a))$ is a monotonically increasing function in $a$, and it follows that there are no finite maximizer for the marginal likelihood $p(\by_n;a,b)$ over $(a,b) \in \mathbb{R}^2_+$.

\subsection{Proof of Proposition 5}\label{pro:prop:finite}
With the hyperpriors $[a]\sim \Gamma(\zeta,\iota)$ and $[b]\propto \mathbf{1}$, the profile maximizer \eqref{eqn:opti} still holds. MMAP then aims to maximize 
\begin{align}\label{eqn:opti:1}
 \tilde{p}( \by_n;a,b^*(a)) = \frac{(a w)^a\Gamma(a+\frac{n-q}{2})(n-q)^{-a} }{\Gamma(a)(a w/(n-q)+w/2)^{a+\frac{n-q}{2}}}\frac{\iota^{\zeta}a^{\zeta-1}\exp(-\iota\cdot a)}{\Gamma(c)}.
\end{align}
Calculating the derivative of logarithm \eqref{eqn:opti:1}, we obtain
$$
\textstyle\nabla_a \log \tilde{p}( \by_n;a,b^*(a)) = \psi(a+\frac{n-q}{2}) - \psi(a) - \log(1+\frac{n-q}{2a}) + \frac{\zeta-1}{a}-\iota,
$$
which is a decreasing function with $\lim_{a\rightarrow \infty} \nabla_a \log  \tilde{p}( \by_n;a,b^*(a))<0$. This guarantees a finite solution for the MMAP optimization problem over $(a,b) \in \mathbb{R}^2_+$.
\end{proof}

\subsection{Proof of Theorem 1}\label{pro:thm:con}
The proof of Theorem 1 requires the following three lemmas. The first lemma provides an upper bound for the RKHS norm of function $f$ for changing scale parameters:

\begin{lemma}[Lemma~4 in \cite{bull2011convergence}]\label{lem:bounded-norm}
If $f \in \mathcal{H}_{\boldsymbol{\theta}}(\Omega)$, then $f \in \mathcal{H}_{\boldsymbol{\theta}'}(\Omega)$ 
  for all $\mathbf{0} < \boldsymbol{\theta}' \le \boldsymbol{\theta}$, and 
  \[
  \textstyle \norm{f}_{\mathcal{H}_{\boldsymbol{\theta}'}(\Omega)}^2 
  \le \left(\prod_{i=1}^d \theta_i/\theta'_i\right) 
  \norm{f}_{\mathcal{H}_{\boldsymbol{\theta}}(\Omega)}^2.
  \]
\end{lemma}

The following two lemmas describe the posterior distribution of $f$ with trend in terms of $\mathcal{H}_{\boldsymbol{\theta}}(\Omega)$. For simplicity, we denote $k_{\bx_i} = K_{\boldsymbol{\theta}}(\bx_i,\cdot)\in\cH_{\boldsymbol{\theta}}(\Omega)$ for $i = 1,...,n$.
\begin{lemma}\label{lem:repre}
  Suppose $f \in \mathcal{H}_{\boldsymbol{\theta}}(\Omega)$, and $\bp(\bx)\in \mathcal{H}_{\boldsymbol{\theta}}(\Omega)$. Let $g(\bx)=f(\bx) -\bp(\bx)^\top \boldsymbol{\beta}$. Then the estimates $\hat{f}_n(\bx)$ and $\hat{\boldsymbol{\beta}}_n$ in Lemma 1 solves the following optimization problem:
      \begin{align}\label{eqn:equ}
          &\min_{\boldsymbol{\beta}, g(\bx)} \norm{g}_{\mathcal{H}_{\boldsymbol{\theta}}(\Omega)}^2\\ 
      \text{subject to}& \quad  \bp(\bx_i)^\top\boldsymbol{\beta}+ g(\bx_i) = f(\bx_i), \quad i = 1, \cdots, n,\notag
      \end{align}
     with minimum value $n \hat{\sigma}_n^2 $.
\end{lemma}
\begin{proof}
    Since $\Omega$ is a compact domain, $g(\bx) = f(\bx) - \bp(\bx)^\top \boldsymbol{\beta}$, still belongs to the space $\cH_{\boldsymbol{\theta}}(\Omega)$. Let $W = \textrm{Span}(k_{\bx_1}, \dots, k_{\bx_n})$, and decompose $g = 
      g^{\parallel} + g^\perp$, where $g^{\parallel} \in W$, 
      $g^\perp \in W^\perp$, the orthogonal complement space of $W$. It follows that $g^\perp(\bx_i) = \langle 
      g^\perp, k_{
      \bx_i} \rangle = 0$, Since $g^\perp$ affects the 
      optimization only through $\norm{g}_{\mathcal{H}_{\boldsymbol{\theta}}(\Omega)}$, the minimizer must satisfy $g^\perp = 0$.
      
      We can now represent $g$ as $g = \sum_{i=1}^n \upsilon_i k_{\bx_i}$, for some $\upsilon_i \in\RR $, $i = 1,...,n$.
      The optimization problem~\eqref{eqn:equ} then becomes
      \[\min_{\boldsymbol{\upsilon},\boldsymbol{\beta}} \boldsymbol{\upsilon}^\top \mathbf{K}_n \boldsymbol{\upsilon} \quad \text{subject to } \quad
      \mathbf{P}_n\boldsymbol{\beta} + \mathbf{K}_n\boldsymbol{\upsilon} = \by_n,\]
which gives the estimates in Theorem 1.
\end{proof}

The third lemma gives a useful upper bound on the difference between the true function $f$ and the GP predictor $\hat{f}_n$:
\begin{lemma}[Theorem 11.4 in \cite{wendland2004scattered}]\label{lem:err}
For $f\in\cH_{\boldsymbol{\theta}}(\Omega)$, the GP predictor $\hat{f}_n$ has the following pointwise error bound:
\begin{align}
|f(\bx) - \hat{f}_n (\bx) |\leq s_n(\bx) \norm{f}_{\cH_{\boldsymbol{\theta}}(\Omega)}.
\end{align}
\end{lemma}

With these lemmas in hand, we now proceed with the proof of Theorem 1:
\begin{proof}
Recall, we denote $I_n(\bx) = y_n^*-\hat{f}_n(\bx)$. For simplicity, we further denote $u_n(\bx) = (y_n^* - f(\bx))_+$ and $$\tau_n(x) = x\Phi_{\nu_n}(x)+\sqrt{\frac{\nu_n}{\nu_n-2}}\cdot \phi_{\nu_n-2}\left(\frac{x}{\sqrt{\nu_n/(\nu_n-2)}}\right).$$
The HEI criterion can then be written as:
\begin{align}
\label{eqn:heipf}
\textstyle\HEI_n(\bx) = \tilde{\sigma}_ns_n(\bx)\tau_n\left(\frac{I_n(\bx)}{\tilde{\sigma}_ns_n(\bx)}\right). 
\end{align}
Since $\tau_n'(x) = \Phi_{\nu_n}(x)\geq 0$, $\tau_n(x)$ must be non-decreasing in $x$. Moreover, we denote $R = \norm{f}_{\cH_{\boldsymbol{\theta}^U}(\Omega)}$, then by Lemma~\ref{lem:bounded-norm}, we have $\norm{f}_{\mathcal{H}_{\tilde{\boldsymbol{\theta}}_n}(\Omega)}\leq R\sqrt{\prod_{i=1}^d\frac{\theta^U_i}{\theta^L_i}}:=S$, and by Lemma~\ref{lem:err}, if $u_n(\bx)>0$, then $|u_n(\bx)-I_n(\bx)|\leq s_n(\bx) S$. Thus, 
\begin{align}\label{eqn:EI1}
\textstyle\HEI_n(\bx)\geq \tilde{\sigma}_ns_n(\bx)\tau_n\left(\frac{u_n(\bx)-Ss_n(\bx)}{\tilde{\sigma}_ns_n(\bx)}\right)\geq \tilde{\sigma}_ns_n(\bx)\tau_n\left(\frac{-S}{\tilde{\sigma}_n}\right).
\end{align}
Note that $\tau_n(x) = x - x\Phi_{\nu_n}(-x)+ \sqrt{\nu_n/(\nu_n-2)} \phi_{\nu_n-2}(x/\sqrt{\nu_n/(\nu_n-2)}) = x + \tau_n(-x).$ Therefore,
\begin{align}\label{eqn:EI2}
\textstyle \HEI_n(\bx) \geq \tilde{\sigma}_ns_n(\bx) \tau_n\left(\frac{u_n(\bx)-Ss_n(\bx)}{\tilde{\sigma}_ns_n(\bx)}\right)\geq u_n(\bx)-Ss_n(\bx).
\end{align}
On the one hand, by inequalities~\eqref{eqn:EI1} and \eqref{eqn:EI2}, we have the following lower bound on $\textstyle\HEI_n(\bx)$:
\begin{align}\label{eqn:bounded}
\textstyle\HEI_n(\bx)\geq \frac{\tilde{\sigma}_n \tau_n(-S/\tilde{\sigma}_n)}{S+\tilde{\sigma}_n \tau_n(-S/\tilde{\sigma}_n)}u_n(\bx) \geq \frac{\tau_n(-S/\tilde{\sigma}_n)}{\tau_n(S/\tilde{\sigma}_n)}u_n(\bx).
\end{align}
Note the above inequality also holds for $u_n(\bx) =0$. Therefore it holds for any $u_n(\bx).$

On the other hand, note that $\tau_n(x) = x+\tau_n(-x)\leq x+\tau_n(0)$ for any $x\geq 0$. Moreover $\tau_n(0) = \sqrt{\nu_n/(\nu_n-2)}\phi_{\nu_n-2}(0)\leq 2$, since $n>q+1$, $\nu_n/(\nu_n-2)\leq 3$ and $\phi_{\nu_n}(0)\leq \phi(0)\leq 2/5$. Thus, $\tau(x)\leq x+2$ for $x\geq 0$. Plugging this into \eqref{eqn:heipf}, we get the following upper bound on $\textstyle\HEI_n(\bx)$:
\begin{align}\label{eqn:boundup}
\textstyle\HEI_n(\bx)\leq \tilde{\sigma}_ns_n(\bx)\tau_n\left(\frac{u_n(\bx)+Ss_n(\bx)}{\tilde{\sigma}_ns_n(\bx)}\right)\leq u_n(\bx)+\left(S +2\tilde{\sigma}_n \right)s_n(\bx).
\end{align}

By Lemma 7 of Bull (2011), we know that there exists a constant $C_2$, depending on $\Omega$, $K$, and $\boldsymbol{\theta}^L$ such that for any sequence $\bx_n\in\Omega$ and $k\in \mathbb{N}$, the inequality $$   \textstyle s_n(\bx_{n+1})\geq C_2 k ^{-(\nu \wedge 1)/d} (\log k)^\zeta$$ holds at most $k$ times.

Furthermore, since $\norm{f}_{\mathcal{H}_{\boldsymbol{\theta}^U}(\Omega)}= R$, we have
$$\textstyle\sum_{i=1}^{n} (y_i^*-y_{i+1}^*) = y_1^*-y_{n+1}^*\leq y_1^* - \min_{\bx\in\Omega} f(\bx)\leq 2\norm{f}_{\infty}\leq 2R.$$

Therefore, by $y_n^*-y_{n+1}^*\geq 0$, it follows that $y_n^*-y_{n+1}^*\geq 2Rk^{-1}$ holds at most $k$ times. Otherwise, the above inequality does not hold. Furthermore, by $y_{n+1}^*\leq f(\bx_{n+1})$, we have $y_n^*- f(\bx_{n+1})\geq 2Rk^{-1}$ holds at most $k$ times. Thus, there exists an $n_k\in\mathbb{N}$, with $k\leq n_k\leq 3k$, for which $$\textstyle s_{n_k}(\bx_{n_k+1})\leq C_2 k ^{-(\nu \wedge 1)/d} (\log k)^\zeta\quad \textrm{and}\quad u_{n_k}(\bx_{n_k+1})\leq 2Rk^{-1}.$$ 
Since $y_n^*$ is non-increasing in $n$, for $3k\leq n < 3(k + 1)$, we further have
\begin{align}
    y_n^* -f(\bx^*) 
    &\le y_{n_k}^* - f(\bx^*)\\ 
    &\le \frac{\tau_{n_k}(S/\tilde{\sigma}_{n_k})}{\tau_{n_k}(-S/\tilde{\sigma}_{n_k})} \mathrm{HEI}_{n_k}(\bx^*)\\
    &\le \frac{\tau_{n_k}(S\tilde{\sigma}_{n_k})}{\tau_{n_k}(-S\tilde{\sigma}_{n_k})} \mathrm{HEI}_{n_k}(\bx_{n_k+1})\\
   &\le \frac{\tau_{n_k}(S/\tilde{\sigma}_{n_k})}{\tau_{n_k}(-S/\tilde{\sigma}_{n_k})} \left( 2Rk^{-1} + C_2\left(S +2\tilde{\sigma}_{n_k} \right)k^{-(\nu \wedge 1)/d}(\log k)^{\zeta}\right)\\
   &\le \frac{\tau_{n_k}(S/c)}{\tau_{n_k}(-S/c)} \left( 2Rk^{-1} + C_3Sk^{-(\nu \wedge 1)/d}(\log k)^{\zeta}\right), 
 \end{align}
where the second last inequality holds from Lemma~\ref{lem:repre} and the last inequality holds from Lemma~\ref{lem:bounded-norm} since $\tilde{\sigma}_{n_k}$ is based on the MAP estimate $\tilde{\boldsymbol{\theta}}_n$. From this, we obtain the desired result
\begin{align}
\textstyle f(\bx_n^*) - \min_{\bx} f(\bx) = \left\{
\begin{array}{ll}
\mathcal{O}(n^{-\nu/d}(\log n)^\alpha),& \nu\leq 1,\\
\mathcal{O}(n^{-1/d}),& \nu> 1.\\
\end{array}
\right.
\end{align}
\end{proof}

\subsection{Proof of Theorem 2}\label{pro:thm:stab}
With Condition 1, the picked points are quasi-uniform distributed by Theorems 14 and 18 in \cite{wenzel2020novel}. Therefore, with large enough $n$, the fill distance can be small enough. Then by Theorem 1 in \cite{wynne2020convergence}, we obtain the desired results. 

\end{appendix}

\end{document}